\documentclass[10pt,final]{article}

\usepackage{a4wide}
\usepackage[utf8]{inputenc}
\usepackage{amsfonts}
\usepackage{amsmath, amsthm}
\usepackage{enumitem}
\usepackage[numbers]{natbib}

\theoremstyle{definition}

\newtheorem{theorem}{Theorem}
\newtheorem{lemma}[theorem]{Lemma}
\newtheorem{proposition}[theorem]{Proposition}
\newtheorem{example}[theorem]{Example}
\newtheorem{definition}[theorem]{Definition}
\newtheorem{remark}[theorem]{Remark}

\newcommand{\A}{\mathcal{A}}
\newcommand{\B}{\mathcal{B}}
\renewcommand{\O}{\mathcal{O}}
\newcommand{\Q}{\mathcal{Q}}
\newcommand{\<}{\langle}
\renewcommand{\>}{\rangle}
\newcommand{\vars}{\text{\upshape{vars}}}
\newcommand{\consts}{\text{\upshape{consts}}}
\newcommand{\len}{\text{len}}
\newcommand{\hpsi}{\widehat{\psi}}
\newcommand{\hchi}{\widehat{\chi}}
\newcommand{\subst}[1]{\bigl[#1\bigr]}
\newcommand{\semequiv}{\mathrel{|}\joinrel\Relbar\joinrel\mathrel{|}}
\newcommand{\fU}{\mathfrak{U}}
\newcommand{\fa}{\mathfrak{a}}
\newcommand{\fb}{\mathfrak{b}}
\newcommand{\fc}{\mathfrak{c}}
\newcommand{\fd}{\mathfrak{d}}
\newcommand{\vu}{\vec{\mathbf{u}}}
\newcommand{\vv}{\vec{\mathbf{v}}}
\newcommand{\vx}{\vec{\mathbf{x}}}
\newcommand{\vy}{\vec{\mathbf{y}}}
\newcommand{\vz}{\vec{\mathbf{z}}}
\newcommand{\va}{\vec{\fa}}
\newcommand{\vb}{\vec{\fb}}
\newcommand{\vc}{\vec{\fc}}
\newcommand{\vd}{\vec{\fd}}
\newcommand{\fP}{\mathfrak{P}}
\newcommand{\hSF}{h_\text{SF}}
\newcommand{\mCNF}{{m_\text{CNF}}}
\newcommand{\mDNF}{{m_\text{DNF}}}
\newcommand{\kappaCNF}{{\kappa_\text{CNF}}}
\newcommand{\kappaDNF}{{\kappa_\text{DNF}}}
\newcommand{\twoup}[2]{{2^{\uparrow #1}(#2)}}
\newcommand{\Mapsto}{{\mathop{\mapsto}}}

\begin{document}

\title{Deciding First-Order Satisfiability when\\ Universal and Existential Variables are Separated}

\author{
	\begin{tabular}{c}
		Thomas Sturm \\
		\small\textit{
			CNRS\thanks{CNRS, LORIA, UMR 7503, 54506 Vandoeuvre-lès-Nancy, France}, 
			Inria\thanks{Inria, 54600 Villers-lès-Nancy, France}, and 
			University of Lorraine\thanks{University of Lorraine, LORIA, UMR 7503, 54506 Vandoeuvre-lès-Nancy, France}, Nancy, France}\\
	\end{tabular}
	\and
	\begin{tabular}{c}
		Marco Voigt\\
		\small\textit{Max Planck Institute for Informatics, Saarbr\"ucken, Germany, and}\\
		\small\textit{Saarbr\"ucken Graduate School of Computer Science, Saarland University}
	\end{tabular}
	\and
	\begin{tabular}{c}
		Christoph Weidenbach \\
		\small\textit{Max Planck Institute for Informatics, Saarbr\"ucken, Germany}
	\end{tabular}
}	
\date{}
\maketitle

\begin{abstract}
We introduce a new decidable fragment of first-order logic with equality, which strictly generalizes 
two already well-known ones---the Bernays--Sch\"onfinkel--Ramsey (BSR) Fragment and the Mo\-nadic Fragment.
The defining principle is the syntactic separation of universally quantified variables from existentially quantified 
ones at the level of atoms. Thus, our classification neither rests on restrictions on quantifier prefixes 
(as in the BSR case) nor on restrictions on the arity of predicate symbols (as in the monadic case).
We demonstrate that the new fragment exhibits the finite model property and
derive a non-elementary  upper bound on the computing time required for deciding 
satisfiability in the new fragment. 
For the subfragment of prenex sentences with the quantifier prefix $\exists^* \forall^* \exists^*$ the satisfiability problem 
is shown to be complete for NEXPTIME.
Finally, we discuss how automated reasoning procedures can take advantage of our results.
\end{abstract}







\section{Introduction}

The question of whether satisfiability of first-order sentences of a certain syntactic form is decidable or not has a long 
tradition in computational logic. 
Over the decades different dimensions have been introduced along which decidable first-order fragments can be separated from 
undecidable ones. L\"owenheim's pioneering work \cite{Lowenheim1915} shows that confinement 
to unary predicate symbols (i.e.\ to the \emph{Relational Monadic Fragment}) leads to decidability, while 
the set of sentences in which binary predicate symbols may be used without further restriction yields 
a \emph{reduction class} for first-order logic. Another dimension is the confinement to certain 
quantifier prefixes for formulas in \emph{prenex normal form}, e.g.\ to $\exists^* \forall^*$---the 
\emph{Bernays--Sch\"onfinkel (BS) Fragment} \cite{Bernays1928}---consisting of decidable sentences. 
The inclusion or exclusion of the distinguished equality predicate yields yet another possibility. 
L\"owenheim had already considered it for the Relational Monadic Fragment, and \citet{Ramsey1930} extended the 
BS Fragment in this way, leading to the \emph{Bernays--Sch\"onfinkel--Ramsey (BSR) Fragment}, which also remains decidable. 
Already in the nineties, the classification along those dimensions was considered to be mainly completed, 
cf.\ \cite{Dreben1979, Borger1997}. 

In the present paper we introduce another dimension of classification: syntactic separation of 
existentially quantified variables from universally quantified ones at the level of atoms. 
We disallow atoms 
$P(\ldots, x, \ldots, y, \ldots)$ that contain both  a universally quantified variable 
$x$ and an existentially quantified variable $y$ at the same time. 
We call the class of first-order sentences built from atoms with thus separated variables the \emph{Separated Fragment (SF)}.
The Separated Fragment is decidable. It properly 
includes both the BSR Fragment and the Relational Monadic Fragment (monadic formulas without 
non-constant function symbols) without equality and is to the best of our knowledge 
the first known decidable fragment enjoying this property.
Consequently, separation of differently quantified variables constitutes a unifying principle 
which underlies both the BSR Fragment and the Relational Monadic Fragment.
More concretely, our contributions are:
\begin{enumerate}[label=(\roman{*}), ref=(\roman{*})]
	\item We introduce the new decidable first-order fragment SF (Sections~\ref{section:Separation}, \ref{section:SeparatedFragment}). 
		Its sentences can be transformed into equivalent BSR sentences and thus enjoy the finite model property. 
		Our current upper bound on the size of small models is non-elementary in the length of the formula
		(Theorem~\ref{theorem:SFGeneralComplexity}).
	\item For SF formulas in which blocks of existentially quantified variables are pairwise disjoint the size of small models is at most double exponential (Theorem~\ref{theorem:SFMutExclusive}).
	\item For SF sentences of the form $\exists^*\forall^*\exists^* \varphi$, where $\varphi$ is quantifier-free, the size of small models is single exponential in the length of the formula. Therefore, satisfiability is NEXPTIME-complete in this setting (Theorem~\ref{theorem:Complexity}).
	\item Sentences from the Monadic Fragment with unary function symbols or from the Relational Monadic Fragment with equality can be translated into SF sentences. The translation preserves satisfiability and increases the length of the formulas only polynomially (Section~\ref{section:FurtherExtensionsSeparated}). 
	\item We provide a methodology for the translation of SF sentences of the form $\exists^*\forall^*\exists^* \varphi$ into the BS(R) fragment, which facilitates automated reasoning (Section~\ref{section:AutomatedReasoning}).
\end{enumerate} 
The paper concludes with a discussion of related and future work in Section~\ref{section:furthrelwork}.

The present paper is an extended version of \cite{Voigt2016}.







\section{Preliminaries}\label{section:Preliminaries}

We consider first-order logic formulas. The underlying signature shall not be mentioned explicitly, but will become clear from the current context. For the distinguished \emph{equality} predicate (whose semantics is fixed to be the identity relation) we use $\approx$. If not explicitly excluded, we allow equality in our investigations.
As usual, we interpret a formula $\varphi$ with respect to given structures. A \emph{structure} $\A$ consists of a nonempty \emph{universe} $\fU_\A$ (also: \emph{domain}) and interpretations $f^\A$ and $P^\A$ of all considered function and predicate symbols, respectively, in the usual way. 
Given a formula $\varphi$, a structure $\A$, and a variable assignment $\beta$, we write $\A, \beta \models \varphi$ if $\varphi$ evaluates to \emph{true} under $\A$ and $\beta$. 
We write $\A \models \varphi$ if $\A, \beta \models \varphi$ holds for every $\beta$.
The symbol $\semequiv$ denotes \emph{semantic equivalence} of formulas, i.e.\ $\varphi \semequiv \psi$ holds whenever for every structure $\A$ and every variable assignment $\beta$ we have $\A,\beta \models \varphi$ if and only if $\A,\beta \models \psi$. 
We call two sentences $\varphi$ and $\psi$ \emph{equisatisfiable} if $\varphi$ has a model if and only if $\psi$ has one.

We use $\varphi(x_1, \ldots, x_m)$ to denote a formula $\varphi$ whose free variables form a subset of $\{x_1, \ldots, x_m\}$.
In all formulas we tacitly assume that no variable occurs freely and bound at the same time and that all quantifiers bind distinct variables.
For convenience, we sometimes identify tuples $\vx$ of variables with the set containing all the variables that occur in $\vx$.

A structure $\A$ is a \emph{substructure} of a structure $\B$ (over the same signature) if (1) $\fU_\A \subseteq \fU_\B$, (2) $c^\A = c^\B$ for every constant symbol $c$, (3) $P^\A = P^\B \cap \fU_\A^m$ for every $m$-ary predicate symbol $P$, and (4) $f^\A(\fa_1, \ldots, \fa_m) = f^\B(\fa_1, \ldots, \fa_m)$ for every $m$-ary function symbol $f$ and every $m$-tuple $\<\fa_1, \ldots, \fa_m\> \in \fU_\A^m$.
The following is a standard lemma, see, e.g., \cite{Ebbinghaus1994} for a proof.

\begin{lemma}[Substructure Lemma]
	Let $\varphi$ be a first-order sentence in prenex normal form without existential quantifiers and let $\A$ be a substructure of $\B$.
	$\B \models \varphi$ entails $\A \models \varphi$.
\end{lemma}

We denote \emph{substitution} by $\varphi\subst{x/t}$, where every free occurrence of $x$ in $\varphi$ is to be substituted by the term $t$. For \emph{simultaneous substitution} of pairwise distinct variables $x_1, \ldots, x_n$ with $t_1, \ldots, t_n$, respectively, we use the notation $\varphi\subst{x_1/t_1, \ldots, x_n/t_n}$. For example, $P(x,y)\subst{x/f(y), y/g(x)}$ results in $P(f(y), g(x))$. We also write $\subst{\vx/\vec{t}\:}$ to abbreviate $\subst{x_1/t_1, \ldots, x_n/t_n}$.

\begin{lemma}[Miniscoping]\label{lemma:BasicQuantifierEquivalences}
	Let $\varphi, \psi, \chi$ be first-order formulas, and assume that $x$ does not occur freely in $\chi$.\\[1ex]
	\noindent
	\begin{tabular}{c@{\;\;\;}l@{\;\;}c@{\;\;}lr}
		(i) 	& $\exists x. (\varphi \vee \psi)$ & $\semequiv$ & $(\exists x_1. \varphi) \vee (\exists x_2. \psi)$ \\
		(ii) 	& $\exists x. (\varphi \circ \chi)$ & $\semequiv$ & $(\exists x. \varphi) \circ \chi$ & with $\circ \in \{\wedge, \vee\}$ \\
		(iii) 	& $\forall x. (\varphi \wedge \psi)$ & $\semequiv$ & $(\forall x_1. \varphi) \wedge (\forall x_2. \psi)$ \\
		(iv) 	& $\forall x. (\varphi \circ \chi)$ & $\semequiv$ & $(\forall x. \varphi) \circ \chi$ & with $\circ \in \{\wedge, \vee\}$
	\end{tabular}
\end{lemma}
Consequently, if $x_1 \not\in \vars(\chi)$ and $x_2 \not\in \vars(\varphi)$ holds for two first-order formulas $\varphi$ and $\chi$, we get $(\exists x_1. \varphi) \wedge (\exists x_2. \chi) \semequiv \exists x_1 x_2. (\varphi \wedge \chi)$ and dually $(\forall x_1. \varphi) \vee (\forall x_2. \chi) \semequiv \forall x_1 x_2. (\varphi \vee \chi)$.







\section{Separated Variables and Transposition of Quantifier Blocks}\label{section:Separation}

Let $\varphi$ be a first-order formula. We call two disjoint sets of variables $\vx$ and $\vy$ \emph{separated in $\varphi$} if and only if for every atom $A$ occurring in $\varphi$ we have $\vars(A) \cap \vx = \emptyset$  or $\vars(A) \cap \vy = \emptyset$.

We first show how we can transpose quantifier blocks if the variables they bind are separated. 
Throughout this section we admit equality or other predicates with fixed semantics in the formulas we consider.
Moreover, function symbols may occur, even in an arbitrarily nested fashion.

\begin{proposition}\label{proposition:TransposingQuantifierBlocks}
	Let $\varphi(\vx, \vy, \vz)$ be a quantifier-free first-order formula in which $\vx$ and $\vy$ are separated. There exists some $m \geq 1$ and a quantifier-free first-order formula $\varphi'(\vx, \vy_1, \ldots, \vy_m, \vz)$ such that $\forall\vx \exists\vy. \varphi(\vx, \vy, \vz)$ and $\exists\vy_1\ldots\exists\vy_m\forall\vx.\varphi'(\vx, \vy_1, \ldots, \vy_m, \vz)$ are semantically equivalent, and the length of each of the tuples $\vy_k$ is identical to the length of $\vy$.	
	Moreover, all atoms in $\varphi'$ are variants of atoms in $\varphi$, i.e.\ for every atom $A'$ in $\varphi'$ there is an atom $A$ in $\varphi$ that is identical to $A'$ up to renaming of variables.
\end{proposition}
\begin{proof}
	We first transform the matrix $\varphi$ into a disjunction of conjunctions of literals. 
	Since the sets $\vx$ and $\vy$ are separated in $\varphi$, the literals in every conjunction can be grouped into three parts:\\
	\strut\;\;(1)\;\;$\psi_i(\vx, \vz)$, containing none of the variables from $\vy$,\\
	\strut\;\;(2)\;\;$\chi_i(\vy, \vz)$, containing none of the variables from $\vx$, and\\
	\strut\;\;(3)\;\;$\eta_i(\vz)$, containing neither variables from $\vx$ nor from $\vy$.
	
	Employing Lemma~\ref{lemma:BasicQuantifierEquivalences}, we move the existential quantifier block $\exists \vy$ inwards such that it only binds the (sub-)conjunctions $\chi_i(\vy,\vz)$.
	The emerging subformulas $\bigl(\exists \vy. \chi_i(\vy,\vz)\bigr)$ shall be treated as indivisible units in the further process.
	
	After moving the $\exists \vy$ block inwards, the resulting formula has the form $\forall \vx. \varphi''$. Next, we transform $\varphi''$ into a conjunction of disjunctions, and move the universal quantifier block $\forall \vx$ inwards in a way analogous to the procedure described for the $\exists \vy$ block. The shape of the result then allows to move the existential quantifiers outwards first (renaming variables where necessary) and the universal ones afterwards so that we again obtain a prenex formula, this time with an $\exists^* \forall^*$ prefix.
	\begin{align*}
		&\forall\vx\, \exists\vy. \varphi(\vx, \vy, \vz)\\
		&\semequiv\; \forall\vx.\exists\vy. \bigvee_i \psi_{i}(\vx, \vz) \wedge\chi_{i}(\vy, \vz) \wedge \eta_{i}(\vz) \\
		&\semequiv\; \forall\vx. \bigvee_i \psi_{i}(\vx, \vz) \wedge \Bigl(\exists\vy. \chi_{i}(\vy, \vz)\Bigr) \wedge \eta_{i}(\vz) \\
		&\semequiv\; \forall\vx. \bigwedge_{k=1}^m \psi'_{k}(\vx, \vz) \vee \bigvee_\ell \Bigl(\exists\vy. \chi'_{k, \ell}(\vy, \vz)\Bigr)  \vee \eta'_{k}(\vz)
	\end{align*}
	\begin{align*}	
		&\semequiv\; \bigwedge_{k=1}^m \Bigl(\forall\vx. \psi'_{k}(\vx, \vz)\Bigr) \vee \Bigl( \exists\vy. \bigvee_\ell \chi'_{k, \ell}(\vy, \vz)\Bigr)  \vee \eta'_{k}(\vz) \\
		&\semequiv\; \exists\vy_1\ldots\exists\vy_m\forall\vx. \bigwedge_{k=1}^m \psi'_{k}(\vx, \vz) \vee \bigvee_\ell \chi'_{k, \ell}(\vy_k, \vz)  \vee \eta'_{k}(\vz)
	\end{align*}
	The subformulas $\psi_{i}$, $\chi_{i}$, $\eta_{i}$, and $\chi'_{k, \ell}$ are conjunctions of literals, and the $\psi'_{k}$ and $\eta'_{k}$ are disjunctions of literals.
\end{proof}
The following example illustrates how quantifiers can be transposed in accordance with Proposition~\ref{proposition:TransposingQuantifierBlocks} and what the limits are.

\begin{example}
	Consider the sentence $\forall x \exists y. P(x) \leftrightarrow Q(y)$. 	
	On the one hand, it is easy to see that direct transposition of $\forall x$ and $\exists y$ does not preserve semantics: 
	whilst the structure $\A$ with $\fU_\A := \{\fa, \fb\}$ and $P^\A := Q^\A := \{\fa\}$ is a model of $\forall x \exists y. P(x) \leftrightarrow Q(y)$, it does not satisfy the version with transposed quantifiers $\exists y \forall x. P(x) \leftrightarrow Q(y)$.

	On the other hand, we can show equivalence of $\forall x \exists y. P(x) \leftrightarrow Q(y)$ and $\exists y_1 y_2 \forall x. \bigl(P(x) \rightarrow Q(y_1)\bigr) \wedge \bigl( Q(y_2) \rightarrow P(x) \bigr)$:
	\begin{align*}
		&\forall x \exists y. P(x) \leftrightarrow Q(y)\\
		&\semequiv\; \forall x \exists y. \bigl(\neg P(x) \vee Q(y)\bigr) \wedge \bigl( P(x) \vee \neg Q(y) \bigr) \\
		&\semequiv\; \forall x. \bigl( \neg P(x) \wedge (\exists y_2.\neg Q(y_2)) \bigr) \vee \bigl((\exists y_1.Q(y_1)) \wedge P(x)\bigr) \\
		&\semequiv\; \bigl( (\forall x.\neg P(x)) \vee (\exists y_1. Q(y_1)) \bigr)\\
			&\hspace{20ex} \wedge \bigl( (\exists y_2.\neg Q(y_2)) \vee (\forall x.P(x)) \bigr) \\
		&\semequiv\; \exists y_1 y_2\forall x. \bigl( \neg P(x) \vee Q(y_1) \bigr) \wedge \bigl( \neg Q(y_2) \vee P(x) \bigr)\\
		&\semequiv\; \exists y_1 y_2 \forall x. \bigl(P(x) \rightarrow Q(y_1)\bigr) \wedge \bigl( Q(y_2) \rightarrow P(x) \bigr) ~.
	\end{align*}
\end{example}

Next, we demonstrate how the transposition of quantifiers affects the length of formulas.
\begin{proposition}\label{proposition:Blowup}
	In the worst case transposing quantifier blocks in accordance with Proposition~\ref{proposition:TransposingQuantifierBlocks} leads to a blow-up in the number of existentially quantified variables that is exponential in the length of the original formula.
\end{proposition}
\begin{proof}
	Consider the following first-order sentence and how we can transpose the quantifier blocks therein.
	$\varphi := \forall x \exists y . (P_1(x) \leftrightarrow Q_1(y)) \wedge \ldots \wedge (P_n(x) \leftrightarrow Q_n(y))$
	can be transformed into the equivalent $\varphi' :=$	
	\begin{align*}
		\exists& \underbrace{y_{\<0,\ldots,0\>} \ldots y_{\<1,\ldots,1\>}}_{2^n \text{ variables}} \forall x.\\
	 		&\bigwedge_{\bar{b} \in \{0,1\}^n}   \Biggl(\biggl( \bigwedge_{\scriptsize\begin{array}{c} 1 \leq i \leq n,\\ b_i = 1 \end{array}} P_i(x) \quad\wedge \bigwedge_{\scriptsize\begin{array}{c} 1 \leq j \leq n,\\ b_j = 0 \end{array}} \neg P_j(x) \biggr) \\
	 		&\quad\longrightarrow 
				\biggl(\bigwedge_{\scriptsize\begin{array}{c} 1 \leq i \leq n,\\ b_i = 1 \end{array}} Q_i(y_{\bar{b}}) \quad\wedge \bigwedge_{\scriptsize\begin{array}{c} 1 \leq j \leq n,\\ b_j = 0 \end{array}} \neg Q_j(y_{\bar{b}})\biggr)\Biggr)
	\end{align*}
	Consider the following model $\A$ of $\varphi$ and $\varphi'$:	
	\begin{align*}
		\fU_\A &:= \bigl\{ \fa_{\bar{b}} \bigm| \bar{b} \in \{0,1\}^n \bigr\} \cup \bigl\{ \fa'_{\bar{b}} \bigm| \bar{b} \in \{0,1\}^n \bigr\} ~,\\
		P_i^\A &:= \bigl\{ \fa_{\<b_1, \ldots. b_n\>} \bigm| b_i = 1 \bigr\} \qquad\text{ for } i = 1, \ldots, n ~, \\
		Q_i^\A &:= \bigl\{ \fa'_{\<b_1, \ldots. b_n\>} \bigm| b_i = 1 \bigr\} \qquad\text{ for } i = 1, \ldots, n ~.
	\end{align*}
	
	We make the following observation and shall keep it in mind: ($*$) removing any of the $\fa'_{\bar{b}}$ with $\bar{b} \neq \<0, \ldots, 0\>$ from $\A$ does not lead to a model of $\varphi$.
	We will argue that any sentence $\varphi_*$ (in prenex normal form), which is semantically equivalent to $\varphi$ and starts with a quantifier prefix of the form $\exists^*\forall^*$, contains at least $2^n-1$ existential quantifiers.
	
	 Let $\varphi_* := \exists y_1 \ldots y_k \forall x_1 \ldots x_\ell. \chi_*$ (with $\chi_*$ being quantifier-free) be a sentence with minimal $k$ that is semantically equivalent to $\varphi$. Since $\A$ is also a model of $\varphi_*$, we know that there is a sequence of elements $\fc_1, \ldots, \fc_k$ taken from the domain $\fU_\A$ such that $\A, [y_1 \Mapsto \fc_1, \ldots, y_k \Mapsto \fc_k] \models \forall x_1 \ldots x_\ell. \chi_* $.
	Consequently, we can extend $\A$ to a model $\A_*$ (over the same domain) of the Skolemized formula $\varphi_{\text{Sk}} := \forall x_1 \ldots x_\ell. \chi_*\subst{y_1/c_1, \ldots, y_k/c_k}$ by adding $c_j^{\A_*} := \fc_j$ for $j = 1, \ldots, k$. On the other hand, every model of the Skolemized formula $\varphi_{\text{Sk}}$ immediately yields a model of $\varphi_*$.
	
	The signature underlying $\varphi_{\text{Sk}}$ comprises exactly the constant symbols $c_1, \ldots, c_k$ and does not contain any other function symbols. 
	Suppose $k < 2^n-1$. Hence, there is some $\fa'_{\bar{b}}$ with $\bar{b} \neq \<0, \ldots, 0\>$ such that for every $j$ it holds $c_j^{\A_*} \neq \fa'_{\bar{b}}$. By the Substructure Lemma, the following substructure $\B$ of $\A_*$ constitutes a model of $\varphi_{Sk}$: $\fU_\B := \fU_{\A_*} \setminus \{\fa'_{\bar{b}}\}$, $P_i^\B := P_i^{\A_*} \cap \fU_\B$ and $Q_i^\B := Q_i^{\A_*} \cap \fU_\B$ for every $i$, and $c_j^\B := c_j^{\A_*}$ for every $j$.
	
	However, then $\B$ must also be a model of both $\varphi_*$ and $\varphi$, since every model of $\varphi_{\text{Sk}}$ is a model of $\varphi_*$, and because we assumed $\varphi_*$ and $\varphi$ to be semantically equivalent.
	This contradicts observation ($*$), and thus we must have $k \geq 2^n-1$.
			
	This evidently shows that $\varphi'$ is (almost) optimal regarding the number of existentially quantified variables it contains.
\end{proof}

It is possible to generalize  Proposition~\ref{proposition:TransposingQuantifierBlocks} to the case of several quantifier alternations as long as all universally quantified variables are separated from all existentially quantified ones.

\begin{lemma}\label{lemma:TransposingMultipleQuantifierBlocks}
	Let $\varphi(\vx_1, \ldots, \vx_n, \vy_1, \ldots, \vy_n, \vz)$ be a quantifier-free first-order formula in which the sets $\vx_1 \cup \ldots \cup \vx_n$ and $\vy_1 \cup \ldots \cup \vy_n$ are separated. 	
	There exists a quantifier-free first-order formula $\varphi'(\vu, \vv, \vz)$ such that $\forall\vx_1 \exists\vy_1 \ldots \forall\vx_n \exists\vy_n. \varphi(\vx_1, \ldots, \vx_n, \vy_1, \ldots,$ $\vy_n, \vz)$ and $\exists \vu \forall \vv. \varphi'(\vu, \vv, \vz)$ are semantically equivalent and all atoms in $\varphi'$ are variants of atoms in $\varphi$.
	Notice that $\vx_1$ and $\vy_n$ may be empty.
\end{lemma}	
\begin{proof}
	We apply a syntactic transformation following the strategy of the proof of Proposition~\ref{proposition:TransposingQuantifierBlocks}, but in an iterated fashion:	
	\begin{align*}
		\forall\vx_1 &\exists\vy_1 \ldots \forall\vx_n \exists\vy_n. \varphi(\vx_1, \ldots, \vx_n, \vy_1, \ldots, \vy_n, \vz)\\
		\semequiv&\; \forall\vx_1 \exists\vy_1 \ldots \forall\vx_n \exists\vy_n. \bigvee_i \psi_{i}(\vx_1, \ldots, \vx_n, \vz) \\[-1ex]
			&\hspace{25ex} \wedge \chi_{i}(\vy_1, \ldots, \vy_n, \vz) \wedge \eta_{i}(\vz) \\
		\semequiv&\; \forall\vx_1 \exists\vy_1 \ldots \forall\vx_n. \bigvee_i \psi_{i}(\vx_1, \ldots, \vx_n, \vz) \\[-1ex]
			&\hspace{15ex} \wedge  \underbrace{\Bigl(\exists\vy_n.\chi_{i}(\vy_1, \ldots, \vy_n, \vz)\Bigr)}_{=:\; \chi^{(1)}_i (\vy_1, \ldots, \vy_{n-1}, \vz)} \wedge\; \eta_{i}(\vz) \\[0.5ex]
		\semequiv&\; \forall\vx_1 \exists\vy_1 \ldots \forall\vx_n. \bigwedge_k \psi'_{k}(\vx_1, \ldots, \vx_n, \vz) \\[-1ex]
			&\hspace{20ex} \vee \chi'^{(1)}_{k} (\vy_1, \ldots, \vy_{n-1}, \vz) \vee \eta'_{k}(\vz) \\
		\semequiv&\; \forall\vx_1 \exists\vy_1 \ldots \exists\vy_{n-1}.\bigwedge_k \underbrace{\Bigl(\forall\vx_n. \psi'_{k}(\vx_1, \ldots, \vx_n, \vz)\Bigr)}_{=:\; \psi^{(1)}_k (\vx_1, \ldots, \vx_{n-1}, \vz)} \\
			&\hspace{20ex} \vee \chi'^{(1)}_{k} (\vy_1, \ldots, \vy_{n-1}, \vz) \vee \eta'_{k}(\vz) \\
		\semequiv&\; \forall\vx_1 \exists\vy_1 \ldots \exists\vy_{n-1}.\bigvee_i \psi'^{(1)}_{i} (\vx_1, \ldots, \vx_{n-1}, \vz) \\[-1ex]
			&\hspace{20ex} \wedge \chi''^{(1)}_{i} (\vy_1, \ldots, \vy_{n-1}, \vz) \wedge \eta''_{i}(\vz) 
	\end{align*}
	\begin{align*}	
		\semequiv&\; \forall\vx_1 \exists\vy_1 \ldots \forall\vx_{n-1}.\bigvee_i \psi'^{(1)}_{i} (\vx_1, \ldots, \vx_{n-1}, \vz) \\[-1ex]
			&\hspace{10ex} \wedge \underbrace{\Bigl( \exists\vy_{n-1}.\chi''^{(1)}_{i} (\vy_1, \ldots, \vy_{n-1}, \vz)\Bigr)}_{=:\; \chi^{(2)}_i (\vy_1, \ldots, \vy_{n-2}, \vz)} \wedge\; \eta''_{i}(\vz) \\[-1.5ex]
		&\quad\vdots\\
		\semequiv&\; \forall\vx_1. \bigwedge_k \psi''^{(n-1)}_{k}(\vx_1, \vz) \vee \chi'^{(n)}_{k} (\vz) \vee \eta^*_{k}(\vz) \\
		\semequiv&\; \bigwedge_k \underbrace{\Bigl( \forall\vx_1. \psi''^{(n-1)}_{k}(\vx_1, \vz)\Bigl)}_{=:\; \psi^{(n)}_k (\vz)} \vee \chi'^{(n)}_{k} (\vz) \vee \eta^*_{k}(\vz) \\
		\semequiv&\; \exists \vu \forall \vv.  \bigwedge_k \hpsi^{(n)}_k (\vu,\vz) \vee \hchi^{(n)}_{k} (\vv,\vz) \vee \eta^*_{k}(\vz) ~.
	\end{align*}
	The $\chi'^{(\ell)}_{k}$ are disjunctions of subformulas $\chi^{(\ell)}_j$, whereas the $\chi''^{(\ell)}_{k}$ are conjunctions of such formulas.
	Dually, the $\psi'^{(\ell)}_{k}$ are conjunctions of subformulas $\psi^{(\ell)}_j$ and the $\psi''^{(\ell)}_{k}$ are disjunctions of such formulas.
	The $\hpsi^{(n)}_k (\vu,\vz)$ and the $\hchi^{(n)}_{k} (\vv,\vz)$ are quantifier-free variants of $\psi^{(n)}_k (\vz)$ and $\chi'^{(n)}_{k} (\vz)$, respectively, i.e.\ quantifiers have been moved outwards and variables have been renamed appropriately.
\end{proof}
Notice that every transformation into a disjunction of conjunctions or a conjunction of disjunctions in the above proof causes at most an exponential increase in the length of the formula. Moving quantifier blocks inwards causes at most a factor of two per moved block. 
Roughly speaking, the length of $\varphi'$ is thus at most $2n$-fold exponential in the double length of $\varphi$.
This constitutes a very crude upper bound, and it does not take into account that redundant subformulas may be removed.
The complexity results in Section~\ref{section:Skolemization} will be based on semantic arguments, and thus do not rely on this rough estimate.







\section{The Separated Fragment of First-Order Logic}\label{section:SeparatedFragment}
We next explore a striking consequence of Lemma~\ref{lemma:TransposingMultipleQuantifierBlocks}. It is well-known that the Bernays--Sch\"on\-finkel--Ramsey (BSR) Fragment is decidable. This fragment comprises all first-order sentences in prenex normal form with quantifier prefixes of the form $\exists^*\forall^*$. Moreover, equality may occur in BSR sentences but non-constant function symbols may not. By virtue of Lemma~\ref{lemma:TransposingMultipleQuantifierBlocks}, we can now extend this decidability result to the following class of first-order sentences.
\begin{definition}[Separated Fragment (SF)]\label{definition:SeparatedFragment}
	The \emph{Separated Fragment (SF)} of first-order logic shall consist of all existential closures of prenex formulas in which existentially quantified variables are separated from universally quantified ones. 
	
	More precisely, it consists of all first-order sentences with equality but without non-constant function symbols that are of the form $\exists \vz\, \forall\vx_1 \exists\vy_1 \ldots \forall\vx_n \exists\vy_n. \varphi(\vx_1, \ldots, \vx_n, \vy_1, \ldots, \vy_n,$ $\vz)$ in which $\varphi$ is quantifier-free and the sets $\vx_1 \cup \ldots \cup \vx_n$ and $\vy_1 \cup \ldots \cup \vy_n$ are separated.
\end{definition}
As already stated, Lemma~\ref{lemma:TransposingMultipleQuantifierBlocks} shows that every sentence in SF can be transformed into an equivalent one which belongs to the BSR Fragment: just replace the subformula $\forall\vx_1 \exists\vy_1 \ldots \forall\vx_n \exists\vy_n. \varphi$ with an equivalent one of the form $\exists \vu \forall \vv. \varphi'$. This proves our main result:
\begin{theorem}\label{theorem:DecidabilitySF}
	Satisfiability of sentences in SF is decidable.
\end{theorem}

It is worth noticing that SF is not only a proper superset of the BSR Fragment but also of one more well-known decidable first-order fragment, namely the Relational Monadic Fragment \emph{without equality}, i.e.\ the set of first-order sentences without non-constant function symbols in which only predicate symbols of arity one are allowed.

\begin{theorem}\label{theorem:InclusionSF}
	SF properly contains the BSR fragment and the Relational Monadic Fragment without equality.
\end{theorem}
\begin{proof}
	Let $\varphi := \exists \vz\, \forall \vx.\psi$ be a BSR sentence, i.e.\ it may contain equality but non-constant function symbols may not appear. Moreover, assume that $\psi$ is quantifier free.
	The sentence $\varphi$ may be considered as the existential closure of the formula $\forall \vx.\psi$. Since $\forall \vx. \psi$ does not contain any existentially quantified variables, the separation criterion in Definition~\ref{definition:SeparatedFragment} is trivially fulfilled. Hence, $\varphi$ lies in SF.
	
	Let $\varphi'$ be a relational monadic sentence without equality and without non-constant function symbols. 
	Since all predicate symbols in $\varphi'$ have an arity of at most one, any two disjoint sets of variables are trivially separated in $\varphi'$.
	Therefore, $\varphi'$ belongs to SF.
\end{proof}

In Section \ref{section:FurtherExtensionsSeparated} we show how SF can be extended so that the Relational Monadic Fragment \emph{with equality} and also the Full Monadic Fragment \emph{without equality} (i.e.\ Relational Monadic plus unary function symbols) become proper subsets of the extension.

\subsection{Range-Restricted Skolemization}\label{section:Skolemization}

In this section we shall demonstrate another approach to showing decidability of SF. This approach emphasizes the \emph{small model property} of SF, i.e.\ we can give a computable function $\hSF$ which takes any SF sentence $\varphi$ as input and yields a positive integer $\hSF(\varphi)$ such that whenever $\varphi$ has a model then it also has a model based on a universe with at most $\hSF(\varphi)$ elements.

It is well-known that BSR sentences exhibit the small model property, where the size of small models is linear in the number of occurring constant symbols plus the number of occurrences of existential quantifiers in the sentence at hand. 
Hence, the number of constant symbols and existentially quantified variables in any BSR sentence $\varphi'$ that is semantically equivalent to an SF sentence $\varphi$ would yield an upper bound on the size of small models of $\varphi$. However, the transformations carried out in the proof of Lemma~\ref{lemma:TransposingMultipleQuantifierBlocks} do not immediately admit a reasonable estimate on the number of variables in the resulting BSR sentence. This is why we tackle the problem in a different way. As it turns out, this alternative approach does not only facilitate the derivation of tighter upper bounds on the size of small models in subfragments of SF. In addition, automated reasoning procedures may benefit from the developed methods and results. We shall assess this potential in Section~\ref{section:AutomatedReasoning}.

On an abstract level, our semantic approach is akin to proofs of the small model property of relational monadic sentences. Usually, the central argument goes as follows: any sentence $\varphi$ without equality and without non-constant function symbols which contains exactly $k$ predicate symbols $P_1, \ldots, P_k$---all of them unary---cannot distinguish more than $2^k$ domain elements. To formalize this, we associate with every domain element $\fa \in \fU_\A$ of a given structure $\A$ a \emph{fingerprint} with respect to the predicates $P_1^\A, \ldots, P_k^\A$, namely the set $\lambda(\fa) := \{i \mid \fa \in P_i^\A\}$. Hence, $\lambda(\fa)$ contains exactly the indices $i$, $1\leq i\leq k$, for which $\fa$ belongs to $\A$'s interpretation of $P_i$. In a certain sense, this fingerprint of a domain element is all that matters for a relational monadic sentence $\varphi$ under $\A$, i.e., if two elements $\fa, \fb$ have the same fingerprint $\lambda(\fa) = \lambda(\fb)$, then $\varphi$ cannot distinguish the two. As a consequence, given a model of $\varphi$, domain elements with identical fingerprints can be merged, and the resulting structure is still a model of $\varphi$. Since there are at most $2^k$ distinct fingerprints, this entails the small model property for relational monadic sentences.

In order to treat SF sentences, we need to modify the just described idea of fingerprints in several ways:

(a) Since the definition of SF does not pose any restriction on the arity of predicate symbols, we have to generalize the idea of a fingerprint from single elements to tuples of elements.

(b) We shall concentrate on the parts of sentences which contain existentially quantified variables---by definition of SF these can be isolated from the ones containing universally quantified variables.
	The rationale behind this approach is rooted in the idea underlying the Substructure Lemma, namely that only the part of the domain is of interest, which is generated by the interpretations of function symbols. And since in the SF setting non-constant function symbols are only introduced by Skolemization, we focus on existentially quantified variables.

(c) Instead of regarding the membership of a tuple $\va$ in predicates $P^\A \subseteq \fU_\A^m$ as the characteristic feature to define fingerprints, we consider whether $\va$ satisfies certain subformulas of normal forms of $\varphi$. More precisely, we refer to the \emph{disjunctive normal form (DNF)} $\bigvee_i \psi_i$ or the \emph{conjunctive normal form (CNF)} $\bigwedge_j \psi'_j$ and ask the question whether the part $\eta_i(\vy)$ of a conjunction (or disjunction) $\psi_i = \chi_i(\vx) \wedge \eta_i(\vy)$ is true under $\A, [\vy \Mapsto \va\,]$ or not. If (and only if) it is, then the index $i$ belongs to the fingerprint of $\va$.

(d) In general settings with several quantifier alternations we cannot only rely on a single fingerprint function, but rather construct one for every existential quantifier block occurring in the quantifier prefix.
The fingerprints associated with the earlier quantifier blocks will be nested in the sense that they comprise all the potential fingerprints that may be produced starting from the current point.
For example, consider the sentence $\forall x_1 \exists y_1 \forall x_2 \exists y_2. \bigl(Q_1(x_1, x_2) \wedge R_1(y_1, y_2) \bigr) \vee \bigl(Q_2(x_1, x_2) \wedge R_2(y_1, y_2)\bigr)$ and the structure $\A$ with $\fU_\A := \{\fa, \fa', \fb\}$ and 
$R_1^\A := \{\<\fa, \fb\>, \<\fa',\fa\>, \<\fa',\fa'\>, \<\fb, \fb\>\}$, 
$R_2^\A := \{\<\fa, \fa\>, \<\fa, \fa'\>, \<\fa, \fb\>,$ $\<\fa', \fa\>, \<\fa', \fa'\>, \<\fa', \fb\>, \<\fb, \fa\>\}$.
($Q_1^\A$ and $Q_2^\A$ may be defined arbitrarily.)
First, we define a fingerprint function $\lambda_2$ for pairs of elements from $\fU_\A$ such that for every pair $\<\fc, \fd\> \in \fU_\A^2$ and every index $i = 1,2$ we have $i \in \lambda_2(\fc, \fd)$ if and only if $\A, [y_1 \Mapsto \fc, y_2 \Mapsto \fd] \models R_i(y_1, y_2)$. 
Concretely, $\lambda_2$ assigns fingerprints as follows:
\begin{center}
$
\begin{array}{c|ccc}
	\lambda_2	&	\fa		&	\fa'		&	\fb \\ 
	\hline\\[-1.5ex]
	\fa		&	\{2\}		&	\{2\}		&	\{1,2\} \\
	\fa'		&	\{1,2\}	&	\{1,2\}	&	\{2\} \\
	\fb		&	\{2\}		&	\emptyset	&	\{1\} \\
\end{array}
$
\end{center}

Based on $\lambda_2$ we next define the fingerprint function $\lambda_1$ for single domain elements such that for every element $\fc \in \fU_\A$ and every $\lambda_2$-fingerprint $S \subseteq \{1,2\}$ it holds $S \in \lambda_1(\fc)$ if and only if there is another element $\fd\in \fU_\A$ such that $S = \lambda_2(\fc, \fd)$.
In the above example, this means
$\lambda_1(\fa) = \lambda_1(\fa') = \bigl\{ \{2\}, \{1,2\} \bigr\}$, and
$\lambda_1(\fb) = \bigl\{ \emptyset, \{1\}, \{2\} \bigr\}$.
We will see later that the fact that $\lambda_2$ assigns the same fingerprint to $\fa$ and $\fa'$ entails that 
the quantifier $\exists y_1$ does not have to take both $\fa$ and $\fa'$ into account. It suffices to consider only one of them.
Hence, after Skolemizing $\exists y_1$, we may restrict the range of the Skolem term $f_{y_1}(x_1)$ under $\A$ to $\{\fa, \fb\}$ or $\{\fa', \fb\}$, but we do not have to consider the full range $\{\fa, \fa', \fb\}$.
This is what we call \emph{range-restricted Skolemization} (cf.\ Lemma~\ref{lemma:SkolemizationOneQuantifierBlock} and similar results in Section~\ref{section:SkolemizationSeveralQuantifierBlocks}).

In what follows we shall use the notation $[k]$ to abbreviate the set $\{1, \ldots, k\}$ for any positive integer $k$.
Moreover, $\fP$ shall be used as the power set operator, i.e.\ $\fP S$ denotes the set of all subsets of a given set $S$.

\subsubsection{Sentences with prefix $\exists^* \forall^* \exists^*$}\label{section:SkolemizationSimpleSentence}

We develop the following for sentences with the $\forall^* \exists^*$ quantifier prefix, because leading existential quantifiers may be replaced by constant symbols (under preservation of satisfiability).
Let $\psi(\vx, \vy)$ be a quantifier-free first-order formula without non-constant function symbols in which $\vx$ and $\vy$ are separated. 
We can transform $\varphi := \forall\vx \exists\vy. \psi(\vx, \vy)$ into an equivalent sentence $\varphi_{\text{DNF}} := \forall\vx \exists\vy. \bigvee_{k=1}^\mDNF \chi_k(\vx) \wedge \eta_k(\vy)$, such that $\chi_k$ and $\eta_k$ are conjunctions of literals and $\mDNF$ is some non-negative integer. 

Intuitively speaking, we combine two complementary fingerprint functions in this setting, where the second one will only appear implicitly.
Given a structure $\A$, for any tuple $\vb \in \fU_\A^{|\vx|}$ of domain elements assigned to $\vx$ it is only of importance which subformulas $\chi_k(\vx)$ it fulfills. 
Hence, the fingerprint function $\lambda(\vb) := \bigl\{ k \in [\mDNF] \bigm| \A,[\vx\Mapsto \vb] \models \chi_k(\vx) \bigr\}$ is a reasonable choice.
Complementary to every $\vb$ a tuple $\va \in \fU_\A^{|\vy|}$ is necessary, for which $\A, [\vy\Mapsto \va\,] \models \eta_k(\vy)$ holds for at least one $k \in \lambda(\vb)$. Note that empty $\eta_k$ are treated as the Boolean constant \emph{true}. 
Together, two such matching tuples satisfy the formula in the structure $\A$.
In order to formally classify tuples $\va$ in accordance with the fingerprint feature just described, we define the sets $\widehat{\fU}_1, \ldots, \widehat{\fU}_\mDNF \subseteq \fU_\A^{|\vy|}$ by $\widehat{\fU}_k := \bigl\{ \va \in \fU_\A^{|\vy|} \bigm| \A, [\vy \Mapsto \va\,] \models \eta_k \bigr\}$ for $k = 1, \ldots, \mDNF$. 
Invoking the axiom of choice, we pick one representative $\alpha_k \in \widehat{\fU}_k$ for every nonempty $\widehat{\fU}_k$ and fix it.
In addition, we define a mapping $\mu : \fU_\A^{|\vx|} \to [\mDNF]$ such that $\mu(\vb)$ yields some fixed index $k \in \lambda(\vb)$, for which $\widehat{\fU}_k$ is nonempty. 
If no such $k$ exists for $\vb$, $\mu(\vb)$ shall be undefined.
If $\A \models \forall\vx \exists\vy. \psi(\vx, \vy)$, then for every $\vb \in \fU_\A^{|\vx|}$ the value $\mu(\vb)$ must be defined, and we have $\A, [\vx \Mapsto \vb, \vy \Mapsto \alpha_{\mu(\vb)}] \models \psi(\vx, \vy)$, by definition of the representative $\alpha_{\mu(\vb)}$.
This means, no matter which values are assigned to the variables in $\vx$, for the valuation of the variables in $\vy$ it is sufficient to consider exclusively values from $\{ \alpha_k \mid k\in [\mDNF]\}$.
Hence, we can restrict the choice for the $\vy$ to the representatives $\alpha_1, \ldots, \alpha_{\mDNF}$.
Therefore, the sentences $\forall\vx \exists\vy. \psi(\vx, \vy)$ and $\forall\vx \exists\vy. \psi(\vx, \vy) \wedge \bigvee^{\mDNF}_{\ell=1} \bigwedge_{i=1}^{|\vy|} y_i \approx c_{\ell,i}$ are equisatisfiable, where the $c_{\ell,i}$ are fresh constant symbols.
If there is a model at all, then there is one, under which the tuples of constant symbols $\vec{c}_{1}, \ldots, \vec{c}_{\mDNF}$ are interpreted by the representatives $\alpha_1, \ldots, \alpha_{\mDNF}$, respectively.

\begin{remark}\label{remark:SkolemizationSimplerProof}
	For this simple setting there is also a shorter argument leading to an even stronger result. We have used the more complicated one, however, since it is easily applicable in the more general case in Section~\ref{section:SkolemizationSeveralQuantifierBlocks}, where the simple argument does not work in a straight-forward fashion.
	
	Due to Lemma~\ref{lemma:BasicQuantifierEquivalences}, $\varphi_{\text{DNF}} = \forall\vx \exists\vy. \bigvee_{k=1}^\mDNF \chi_k(\vx) \wedge \eta_k(\vy)$ is equivalent to $\forall\vx. \bigvee_{k=1}^\mDNF \chi_k(\vx) \wedge \exists\vy_k. \eta_k(\vy_k)$. Now \emph{inner Skolemization} (cf.\ \cite{NonnengartWeidenbach01handbook}) leads to $\forall\vx. \bigvee_{k=1}^\mDNF$ $\chi_k(\vx) \wedge \eta_k(\vy_k)\subst{\vy_k/\vc_k}$. This immediately entails equisatisfiability of the latter sentence and $\varphi_{\text{DNF}}$.
\end{remark}

Dually, we can transform the sentence $\varphi = \forall\vx \exists\vy. \psi(\vx, \vy)$ into an equivalent one $\varphi_{\text{CNF}} := \forall\vx \exists\vy. \bigwedge^{m_\text{CNF}}_{j=1} \chi'_j(\vx) \vee \eta'_j(\vy)$ in which the $\chi'_j$ and $\eta'_j$ denote disjunctions of literals and $\mCNF$ is some non-negative integer.
Given a structure $\A$, we define the fingerprint function $\lambda : \fU_\A^{|\vy|} \to \fP [\mCNF]$ by setting $\lambda(\va) := \bigl\{ j \bigm| \A,[\vy\Mapsto \va\,] \models \eta'_j(\vy) \bigr\}$. 
The fingerprints assigned by $\lambda$ induce a partition of $\fU_\A^{|\vy|}$ into at most $2^\mCNF$ equivalence classes:
two tuples $\va, \va'$ are equivalent if and only if they have the same fingerprint $\lambda(\va) = \lambda(\va')$.
For every fingerprint $S \subseteq [\mCNF]$, for which some tuple $\va$ with $\lambda(\va) = S$ exists, we fix a representative $\alpha_S \in \fU_\A^{|\vy|}$ of $\va$\,'s equivalence class, i.e.\ $\lambda(\alpha_{\lambda(\va)}) = \lambda(\va)$.

Clearly, two tuples $\va$, $\va'$ are indistinguishable by the formulas $\eta'_j(\vy)$, if they are associated with the same fingerprint. 
Analogous to our previous result, we could immediately derive equisatisfiability of $\varphi$ and the formula $\forall \vx \exists \vy. \psi(\vx, \vy) \wedge \bigvee_{j=1}^{2^\mCNF} \bigwedge_{i=1}^{|\vy|} y_i \approx c_{j,i}$, where the $c_{j,i}$ are fresh.

However, it turns out that being indistinguishable is more than we need.
In fact, even in the worst case we do not need to consider $2^\mCNF$ different representatives. 
A tuple $\va \in \fU_\A^{|\vy|}$ \emph{covers} a tuple $\va' \in \fU_\A^{|\vy|}$, denoted $\va \sqsupseteq \va'$, if and only if $\lambda(\va') \subseteq \lambda(\va)$.
Of two representatives $\alpha, \alpha'$ with $\alpha \sqsupseteq \alpha'$ we do actually only need one, namely $\alpha$. 
More formally speaking, it is sufficient to partition $\fU_\A^{|\vy|}$ into parts $\widetilde{\fU}_1, \ldots, \widetilde{\fU}_\kappaCNF$ such that
	(i) in every part $\widetilde{\fU}_\ell$ for all distinct $\va, \va' \in \widetilde{\fU}_\ell$ either $\va \sqsupseteq \va'$ or $\va' \sqsupseteq \va$, and
	(ii) for all distinct parts $\widetilde{\fU}_\ell, \widetilde{\fU}_{\ell'}$ all tuples $\va \in \widetilde{\fU}_\ell$ and $\va' \in \widetilde{\fU}_{\ell'}$ are pairwise non-covering, i.e.\ we neither have $\va \sqsupseteq \va'$ nor $\va' \sqsupseteq \va$.
But now, we have to choose the representatives more carefully: a representative $\widetilde{\alpha}_S \in \widetilde{\fU}_\ell$ has to cover all tuples $\va \in \widetilde{\fU}_\ell$. Formulated differently, we have $\widetilde{\alpha}_{\lambda(\va)} \sqsupseteq \va$ for every $\va$.
Putting things together, we observe for all tuples $\vb \in \fU_\A^{|\vx|}$ and $\va \in \fU_\A^{|\vy|}$ that
$\A, \bigl[\vx \Mapsto \vb, \vy \Mapsto \va\bigr] \models \psi(\vx, \vy)$ 
entails
$\A, \bigl[\vx \Mapsto \vb, \vy \Mapsto \widetilde{\alpha}_{\lambda(\va)}\bigr] \models \psi(\vx, \vy)$.
	
It remains to pinpoint the value of $\kappaCNF$.
Put differently, we need to determine the maximal number of pairwise non-inclusive fingerprints.
The answer is provided by Sperner's Theorem.
\begin{theorem}[Sperner's Theorem \cite{Sperner1928}]\label{theorem:Sperner}
	Let $m$ be a positive integer. Consider the lattice formed by the partially ordered set $\<\fP[m], \subseteq\>$.
	There are $\kappa = {m \choose {\lfloor m/2 \rfloor}}$ sets $M_1, \ldots, M_\kappa \subseteq [m]$ such that for all distinct $i,j$ we have $M_i \not\subseteq M_j$. There are no more than $\kappa$ such sets.
\end{theorem}

As a consequence, $\forall\vx \exists\vy. \psi(\vx, \vy)$ is equisatisfiable to $\forall\vx \exists\vy.$ $\psi(\vx, \vy) \wedge \bigvee^{\kappaCNF}_{j=1} \bigwedge_{i=1}^{|\vy|} y_i \approx c_{j,i}$, where $\kappaCNF := {\mCNF \choose {\lfloor \mCNF/2 \rfloor}}$ and the $c_{j,i}$ are fresh.

Together with the dual case, we observe that it is possible to restrict the range of the quantifier block $\exists \vy$ even further to $m_* := \min\bigl(\kappaCNF, \mDNF\bigr)$ different tuples of domain elements (preserving satisfiability). 
After Skolemization these restrictions apply to the freshly introduced Skolem functions and affect their range.
These results are summarized in the following lemma.

\pagebreak[2]
\begin{lemma}[Range-restricted Skolemization]\label{lemma:SkolemizationOneQuantifierBlock}
	The following sentences are pairwise equisatisfiable:
	\begin{enumerate}[label=(\arabic{*}), ref=(\arabic{*})]
		\item\label{enum:SkolemizationOneQuantifierBlock:I} $\forall\vx \exists\vy. \psi(\vx, \vy)$,
		\item $\forall\vx \exists\vy. \psi(\vx, \vy) \wedge \bigvee^{m_*}_{\ell=1} \bigwedge_{i=1}^{|\vy|} y_i \approx c_{\ell,i}$,
		\item\label{enum:SkolemizationOneQuantifierBlock:III} $\forall\vx. \psi(\vx, \vy)\subst{y_1/f_{1}(\vx), \ldots, y_{|\vy|}/f_{|\vy|}(\vx)}$\\
			$\wedge \bigvee^{m_*}_{\ell=1} \bigwedge_{i=1}^{|\vy|} f_i(\vx) \approx c_{\ell,i}$,
	\end{enumerate}
	where $m_* = \min\Bigl( \mDNF, {\mCNF \choose {\lfloor \mCNF/2 \rfloor}} \Bigr)$ and the $c_{\ell,i}$ are fresh constant symbols and the $f_i$ are Skolem functions of appropriate arity.
\end{lemma}

Please note that the just stated result nicely fits together with Proposition~\ref{proposition:TransposingQuantifierBlocks}, which states the equivalence of $\forall\vx \exists\vy. \psi(\vx,\vy)$ and some sentence $\exists\vy_1 \ldots \vy_{m} \forall\vx. \psi'(\vx,\vy_1, \ldots,$ $\vy_{m})$, which results in $\forall \vx. \psi'(\vx,\vy_1, \ldots, \vy_{m})\subst{\vy_1/\vec{c}_1, \ldots, \vy_{m}/\vec{c}_{m}}$, when existentially quantified variables are Skolemized.

\subsubsection*{Small model property}
Let us call the sentence \ref{enum:SkolemizationOneQuantifierBlock:I} in Lemma~\ref{lemma:SkolemizationOneQuantifierBlock} $\varphi$ and number \ref{enum:SkolemizationOneQuantifierBlock:III} $\varphi'$. 
Employing the Substructure Lemma, we can derive the small model property of SF sentences with the $\exists^* \forall^* \exists^*$ prefix from the equisatisfiability of $\varphi$ and $\varphi'$. Suppose $\varphi$ is a satisfiable sentence belonging to SF. In particular, this means $\varphi$ does not contain any non-constant function symbols.
Let $\B$ be a model of $\varphi'$.  We define the structure $\A$ by $\fU_\A := \bigl\{ \fa \in \fU_\B \mid \text{there is some constant $c$ in $\varphi'$ such that $\fa =$}$ $\text{$c^\B$} \bigr\}$; $c^\A := c^\B$ for every constant $c$ in $\varphi'$; $f^\A(\va) := f^\B(\va)$ for every $m$-ary function symbol $f$ in $\varphi'$ and every tuple $\va \in \fU_\A^m$; $P^\A := P^\B \cap \fU_\A^m$ for every $m$-ary predicate symbol $P$ occurring in $\varphi'$.
Since $\A$ is a substructure of $\B$, it is a model of $\varphi'$. It is easy to see that $\A$ is also a model of $\varphi$. Moreover, we have defined $\A$'s universe so that it contains at most $|\consts(\varphi')| = |\consts(\varphi)| + m_* \cdot |\vy|$ elements. 

We next bound $m_*$ from above in terms of the length of $\varphi$.
Again, consider $\varphi_\text{DNF} = \forall\vx \exists\vy.$ $\bigvee_{k=1}^\mDNF \chi_k(\vx) \wedge \eta_k(\vy)$, which is equivalent to $\varphi$. Without loss of generality, we may assume the following:  (i) the conjunctions $\chi_k \wedge \eta_k$ contain only literals which appear in $\psi$ after transformation into negation normal form, and (ii) there are no distinct indices $k_1, k_2$ such that the set of literals occurring in $\chi_{k_1} \wedge \eta_{k_1}$ is a subset of the set of literals occurring in $\chi_{k_2} \wedge \eta_{k_2}$ (otherwise, $\chi_{k_2} \wedge \eta_{k_2}$ would be redundant and could be removed from the disjunction). 
Consequently, $\mDNF$ is bounded from above by the maximal number of pairwise non-inclusive subsets of the set of all literals occurring in $\varphi$.
Hence, by virtue of Sperner's Theorem, an upper bound for $\mDNF$ is ${{\len(\varphi)} \choose {\lfloor \len(\varphi)/2 \rfloor}} \leq 2^{\len(\varphi)}$, where $\len(\cdot)$ shall denote a reasonable measure of length of formulas (taking into account occurrences of quantifiers, Boolean connectives, variables, and occurrences of predicate, function, and constant symbols; we assume $\len( \varphi \rightarrow \psi ) = \len( \neg \varphi \vee \psi )$ and $\len( \varphi \leftrightarrow \psi ) = \len( (\neg \varphi \vee \psi) \wedge (\varphi \vee \neg \psi) )$).
All in all, we may conclude $|\fU_\A| \leq \len(\varphi) + 2^{\len(\varphi)} \cdot \len(\varphi) \leq 2^{3\cdot \len(\varphi)}$. This means that if $\varphi$ is satisfiable, then it has a model of size at most $2^{3\cdot \len(\varphi)}$.

\subsubsection*{Worst-case time complexity}

Lewis employed the following lemma in \cite{Lewis1980} to find upper bounds on the required time for Bernays--Sch\"onfinkel sentences and relational monadic sentences without equality.

\begin{lemma}[\cite{Lewis1980}, Proposition~3.2]
	\label{lemma:Complexity}
	Let $\varphi$ be a first-order sentence in prenex normal form containing $n$ universally quantified variables.
	The question whether $\varphi$ has a model of cardinality $m$ can be decided nondeterministically in time $p\bigl(m^n \cdot \len(\varphi)\bigr)$ for some polynomial $p$.
\end{lemma}

Together with our previous results and known lower bounds on time complexity (cf.\ \cite{Lewis1980}) this yields the following theorem.

\begin{theorem}\label{theorem:Complexity}
	Satisfiability of sentences in SF with the quantifier prefix $\exists^* \forall^* \exists^*$ is NEXPTIME-complete.
\end{theorem}
\begin{proof}
	Let $\varphi := \exists \vz\, \forall \vx \exists \vy. \psi$ be an SF sentence, in which $\psi$ is quantifier-free.
	Due to previous observations we know that the sentence $\varphi'' := \forall\vx \exists\vy. \psi\subst{\vz/\vec{d}\:} \wedge \bigvee^{m_*}_{\ell=1} \bigwedge_{i=1}^{|\vy|} y_i \approx c_{\ell,i}$ (for Skolem constants $d_1, \ldots, d_{|\vz|}$) has a model (based on a universe with at most $2^{3\cdot \len(\varphi)}$ elements) if and only if $\varphi$ has one. Clearly, every model of $\varphi''$ is also a model of $\varphi$.
	By Lemma \ref{lemma:Complexity}, we can nondeterministically check whether $\varphi$ has a model of size $2^{3\cdot \len(\varphi)}$ in at most $p\bigl( 2^{(3\cdot \len(\varphi)^2)} \cdot \len(\varphi) \bigr)$ computational steps for some polynomial $p$.
	
	\citet{Lewis1980} has shown NEXPTIME-hardness of satisfiability of BS sentences, i.e.\ of $\exists^* \forall^*$ sentences. By Theorem~\ref{theorem:InclusionSF}, these are included in the $\exists^*\forall^*\exists^*$ subfragment of SF. 
\end{proof}
It is worth noting that satisfiability of $\exists^* \forall^3 \exists^*$ sentences, in which variables are not separated, is known to be undecidable for several subcases, see \cite{Borger1997} for references.

\subsubsection{Sentences with several blocks of quantifiers}\label{section:SkolemizationSeveralQuantifierBlocks}
As in the previous section we replace leading existential quantifiers with constant symbols in this section.

\subsubsection*{The special case of strong separation}

Consider a sentence $\varphi := \forall \vx_1 \exists \vy_1 \ldots \forall \vx_n \exists \vy_n. \psi$ for some quanti\-fier-free first-order formula $\psi$ without non-constant function symbols. We assume that the sets $\vy_1, \ldots, \vy_n$ and $\vx_1 \cup \ldots \cup \vx_n$ are all pairwise separated in $\psi$. 
Hence, we can transform $\psi$ into a disjunction of conjunctions of the form
	$\psi_\text{DNF} := \bigvee_{j=1}^{\mDNF} \chi_j(\vx_1, \ldots, \vx_n) \wedge \bigwedge_{k=1}^n \eta_{j,k}(\vy_k)$,
in which the $\chi_j$ and the $\eta_{j,k}$ are conjunctions of literals. 
The additional requirement of $\vy_1, \ldots, \vy_n$ being pairwise separated relieves us from the need to use nested fingerprints, since at the level of atoms every tuple $\vy_k$ occurs isolated from the others, i.e.\ the values assigned to one tuple $\vy_{k_1}$ do not influence the truth values of the subformulas $\eta_{j,k_2}(\vy_{k_2})$ under $\A$ for $k_2 \neq k_1$. Nested fingerprints will be necessary in the general case later on.

Let $\A$ be an arbitrary structure over the signature of $\varphi$. For every index $k \leq n$ we define a fingerprint function 
	$\lambda_k : \fU_\A^{|\vy_k|} \to \fP[\mDNF]$
such that for every tuple $\va_k \in \fU_\A^{|\vy_k|}$ it holds $\lambda_k(\va_k) := \bigl\{ j \bigm| \A, [\vy_k \Mapsto \va_k] \models \eta_{j,k} \bigr\}$.
If two tuples $\va_k$, $\va'_k$ are assigned the same fingerprint by $\lambda_k$, then they result in the same truth value for $\eta_{j,k}(\vy_k)$ under $\A$ , i.e.\ for every $j$ we have $\A,[\vy_k \Mapsto \va_k] \models \eta_{j,k}$ if and only if $\A,[\vy_k \Mapsto \va'_k] \models \eta_{j,k}$. Since the variables in $\vy_k$ do not occur in other subformulas than the $\eta_{j,k}$, we conclude the following for every $k$ and an arbitrary variable assignment $\beta$:
$\A,\beta[\vy_k \Mapsto \va_k] \models \forall \vx_{k+1} \exists \vy_{k+1} \ldots \forall \vx_n \exists \vy_n. \psi_\text{DNF}$ holds if and only if $\A,\beta[\vy_k \Mapsto \va'_k] \models \forall \vx_{k+1} \exists \vy_{k+1} \ldots \forall \vx_n \exists \vy_n. \psi_\text{DNF}$. In other words, $\va_k$ and $\va'_k$ are interchangeable as values for $\vy_k$, whenever they are assigned the same fingerprint by $\lambda_k$.
		
Every $\lambda_k$ induces a partition of $\fU_\A^{|\vy_k|}$ into at most $2^\mDNF$ equivalence classes of tuples with identical fingerprints with respect to $\lambda_k$. 
As we have done before, we can define representatives $\alpha_{k,S}$ for every fingerprint $S \subseteq [\mDNF]$, for which there is some tuple $\va$ with $\lambda_k(\va) = S$.
In the end, in analogy to the simpler case, we may derive equisatisfiability of $\varphi$ and  
	$\forall \vx_1 \exists \vy_1 \ldots \forall \vx_n \exists \vy_n. \psi \wedge \bigwedge_{k=1}^{n} \bigvee^{2^\mDNF}_{j=1} \bigwedge_{i=1}^{|\vy_k|} y_{k,i} \approx c_{k,j,i}$,
where all the $c_{k,j,i}$ are fresh constant symbols.

Using the techniques we have seen earlier, we can improve this result in two ways.
For one thing, the last quantifier block $\exists \vy_n$ does not need to range over $2^{\mDNF}$ tuples of constants, but $\mDNF$ are sufficient, as we have already seen in Section~\ref{section:SkolemizationSimpleSentence}.
Secondly, the $\vy_1, \ldots, \vy_{n-1}$ do not need to range over $2^{\mDNF}$ tuples either, since we can stick to \emph{covering} representatives instead of representatives with \emph{exactly the same fingerprint} and then apply Sperner's Theorem again. Thus, we need to consider at most $\kappaDNF := {{\mDNF}\choose{\lfloor \mDNF/2 \rfloor}} \leq 2^{\mDNF}$ representatives for every $k < n$. 

Consequently, we may conclude that the original $\varphi$ is equisatisfiable to 
$\forall \vx_1 \exists \vy_1 \ldots \forall \vx_n \exists \vy_n. \psi \wedge$ $\bigr( \bigwedge_{k=1}^{n-1} \bigvee^{\kappaDNF}_{j=1} \bigwedge_{i=1}^{|\vy_k|} y_{k,i} \approx c_{k,j,i} \bigr) \wedge \bigl( \bigvee^{\mDNF}_{j=1} \bigwedge_{i=1}^{|\vy_n|} y_{n,i} \approx c_{n,j,i} \bigr)$
where all the $c_{k,j,i}$ are fresh.

Applying a similar analysis as in Section~\ref{section:SkolemizationSimpleSentence} leads to the observation that if $\varphi$ has a model, then it has one with at most $|\consts(\varphi)| + \kappaDNF \cdot \sum_{k=1}^{n-1} |\vy_k| + \mDNF \cdot |\vy_n| \leq 2^{3\cdot 2^{\len(\varphi)}}$ domain elements.

\begin{theorem} \label{theorem:SFMutExclusive}
	Let $\varphi := \exists\vz\, \forall \vx_1 \exists \vy_1 \ldots \forall \vx_n \exists \vy_n. \psi$ be a sentence in SF for some quantifier-free $\psi$ without non-constant function symbols. If the sets $\vy_1, \ldots, \vy_n$ are pairwise separated in $\psi$, then satisfiability of $\varphi$ can be decided nondeterministically in time that is at most double exponential in $\len(\varphi)$.
\end{theorem}

\subsubsection*{The general case}

Consider a sentence $\varphi := \forall \vx_1 \exists \vy_1 \ldots \forall \vx_n \exists \vy_n. \psi$ for some quantifier-free $\psi$ without non-constant function symbols in which the sets $\vx_1 \cup \ldots \cup \vx_n$ and $\vy_1 \cup \ldots \cup \vy_n$ are separated. 
We can transform $\psi$ into a disjunction of conjunctions 
	$\psi_\text{DNF} := \bigvee_{j=1}^{\mDNF}$ $\chi_j(\vx_1, \ldots, \vx_n) \wedge \eta_j(\vy_1, \ldots, \vy_n)$
in which the $\chi_j$ and the $\eta_j$ are conjunctions of literals. 
As we have already announced in the beginning of Section~\ref{section:Skolemization}, we need to deal with a nested form of fingerprints in the general case of several quantifier alternations.
The reason is that the truth values of the $\eta_j$ depend on the values assigned to variables across multiple existential quantifier blocks.

For the following definitions we enhance the power set operator $\fP$ by allowing for iteration: $\fP^k S$ shall denote the $k$-fold application of $\fP$ to $S$.
Let $\A$ be an arbitrary structure over the signature of $\varphi$.
We inductively define fingerprint functions $\lambda_{n-1}, \ldots, \lambda_1$ with signatures $\lambda_k : \fU_\A^{\sum_{j=1}^k |\vy_j|} \to \fP^{n-k} [\mDNF]$ as follows.

\begin{itemize}
	\item $\lambda_{n-1} : \fU_\A^{\sum_{j=1}^{n-1} |\vy_j|} \to \fP [\mDNF]$ 
		such that for all tuples $\va_1 \in \fU_\A^{|\vy_1|}, \ldots, \va_{n-1} \in \fU_\A^{|\vy_{n-1}|}$ and every $j \in [\mDNF]$ it holds
		$j \in \lambda_{n-1}(\va_1, \ldots, \va_{n-1})$ if and only if there exists a tuple $\va_n$ such that $\A, [\vy_1 \Mapsto \va_1, \ldots, \vy_{n-1} \Mapsto \va_{n-1}, \vy_n \Mapsto \va_n] \models \eta_{j}$;
	\item $\lambda_{n-2} : \fU_\A^{\sum_{j=1}^{n-2} |\vy_j|} \to \fP^2 [\mDNF]$
		such that for all tuples $\va_1 \in \fU_\A^{|\vy_1|}, \ldots, \va_{n-2} \in \fU_\A^{|\vy_{n-2}|}$ and every $S \in \fP [\mDNF]$ it holds
		$S \in \lambda_{n-2}(\va_1, \ldots, \va_{n-2})$ if and only if there exists a tuple $\va_{n-1}$ such that $\lambda_{n-1}(\va_1, \ldots, \va_{n-2}, \va_{n-1}) = S$;\\[-3ex]
	\item[] \qquad\vdots
	\item $\lambda_{1} : \fU_\A^{|\vy_1|} \to \fP^{n-1} [\mDNF]$
		such that for every tuple $\va_1 \in \fU_\A^{|\vy_1|}$ and every $S \in \fP^{n-2} [\mDNF]$ it holds
		$S \in \lambda_{1}(\va_1)$ if and only if there exists a tuple $\va_{2}$ such that $\lambda_{2}(\va_1, \va_{2}) = S$.
\end{itemize}
The following lemma expresses that sequences of tuples with identical fingerprints may be interchanged without affecting semantics.

\begin{lemma}
	For every $k < n$ and all tuples $\va_1, \va'_1 \in \fU_\A^{|\vy_1|}, \ldots,$ $\va_{k}, \va'_{k} \in \fU_\A^{|\vy_{k}|}$, $\vb_1 \in \fU_\A^{|\vx_1|}, \ldots, \vb_{k} \in \fU_\A^{|\vx_{k}|}$ if $\lambda_k(\va_1, \ldots, \va_{k}) = \lambda_k(\va'_1, \ldots, \va'_{k})$ then it holds\\
	(i) $\A, [\vx_1 \Mapsto \vb_1, \ldots, \vx_{k} \Mapsto \vb_k, \vy_1 \Mapsto \va_1, \ldots, \vy_k \Mapsto \va_k] \models \forall \vx_{k+1}\exists \vy_{k+1} \ldots \forall \vx_n \exists \vy_n. \psi$\\
	if and only if\\
	(ii) $\A, [\vx_1 \Mapsto \vb_1, \ldots, \vx_{k} \Mapsto \vb_k, \vy_1 \Mapsto \va'_1, \ldots, \vy_k \Mapsto \va'_k] \models \forall \vx_{k+1} \exists \vy_{k+1} \ldots \forall \vx_n \exists \vy_n. \psi$.
\end{lemma}
\begin{proof}
	We proceed by induction from $k=n-1$ to $k=1$.
	For the sake of readability, we abbreviate sequences such as $\vx_1 \Mapsto \vb_1, \ldots, \vx_k \Mapsto \vb_{k}$ by $\vx_{\leq k} \Mapsto \vb_{\leq k}$.
	
	Suppose $k=n-1$. Let $S$ be the value of $\lambda_{n-1}(\va_1, \ldots, \va_{n-1})$ (which we assume to be identical to $\lambda_{n-1}(\va'_1, \ldots,$ $\va'_{n-1})$). By definition of $\lambda_{n-1}$, $S$ is a subset of $[\mDNF]$.	
	Assume (i) holds and let $\vc \in \fU_\A^{|\vx_n|}$ be arbitrary. There must be some index $j \in [\mDNF]$ and some $\vd \in \fU_\A^{|\vy_n|}$ such that
	$\A, [\vx_{\leq n-1} \Mapsto \vb_{\leq n-1}, \vx_{n} \Mapsto \vc\,] \models \chi_j$ and
	$\A, [\vy_{\leq n-1} \Mapsto \va_{\leq n-1}, \vy_n \Mapsto \vd\,] \models \eta_j$. 
	Hence, $j \in S$ and thus there is some $\vd' \in \fU_\A^{|\vy_n|}$ such that $\A, [\vy_{\leq n-1} \Mapsto \va'_{\leq n-1}, \vy_n \Mapsto \vd'] \models \eta_j$. Therefore, (ii) must hold too. Since this argument is completely symmetric, (i) holds if and only if (ii) does.
	
	Suppose that $k < n-1$ and let $S := \lambda_{k}(\va_1, \ldots, \va_{k}) = \lambda_{k}(\va'_1, \ldots,$ $\va'_{k})$.	
	Assume (i) and let $\vc \in \fU_\A^{|\vx_{k+1}|}$ be arbitrary. Then there must be some tuple $\vd \in \fU_\A^{|\vy_{k+1}|}$ such that 
	$\A, [\vx_{\leq k} \Mapsto \vb_{\leq k},$ $\vx_{k+1} \Mapsto \vc, \vy_{\leq k} \Mapsto$ $\va_{\leq k}, \vy_{k+1} \Mapsto \vd\,] \models \forall \vx_{k+2} \exists \vy_{k+2} \ldots \forall \vx_n \exists \vy_n. \psi$.
	Because of $\lambda_{k+1}(\va_1, \ldots, \va_k, \vd) \in S$, there must be some tuple $\vd' \in \fU_\A^{|\vy_{k+1}|}$ such that $\lambda_{k+1}(\va'_1, \ldots,$ $\va'_k, \vd') = \lambda_{k+1}(\va_1, \ldots, \va_k,$ $\vd)$. 
	Hence, by induction, we get $\A, [\vx_{\leq k} \Mapsto \vb_{\leq k}, \vx_{k+1} \Mapsto \vc, \vy_{\leq k} \Mapsto$ $\va'_{\leq k}, \vy_{k+1} \Mapsto$ $\vd'] \models \forall \vx_{k+2} \exists \vy_{k+2} \ldots \forall \vx_n \exists \vy_n. \psi$. Consequently, (ii) must hold, too. Again, by symmetry, (ii) does also entail (i).
\end{proof}
Having defined the functions $\lambda_1, \ldots, \lambda_{n-1}$ and having shown that tuples of elements, which are associated with the same fingerprint, are interchangeable, we can now employ the same methods as before: 
(i)~partition the sets $\fU_\A^{\sum_{j=1}^{k} |\vy_k|}$ in accordance with the fingerprints assigned by the $\lambda_k$, 
(ii)~fix representatives $\alpha_{k,S}^\ell$ for every $k$ and every occurring fingerprint $S$ (the $\ell$ accounts for multiple representatives with fingerprint $S$ at level $k$), 
(iii)~restrict the range of quantifier blocks $\exists \vy_k$ to (parts of) the representatives $\alpha_{k,S}^\ell$.
Step (ii) is slightly more complicated in this setting, because the truth value of the subformulas $\eta_j(\vy_1, \ldots, \vy_n)$ under $\A$ depends on all the values assigned to $\vy_1, \ldots, \vy_n$. 
The consequence is not only a nesting of fingerprints but also a nesting of representatives: 
the $\alpha_{k,S}^\ell$ at level $k$ are extensions of some $\alpha_{k-1,S'}^{\ell'}$ at level $k-1$ with $S \in S'$, respectively. 
More precisely, starting from some representative $\alpha_{k-1,S'}^{\ell'} = \<\va_1, \ldots, \va_{k-1}\>$ with the fingerprint $\lambda_{k-1}(\va_1, \ldots, \va_{k-1}) = S'$, we pick one representative $\alpha_{k,S}^\ell = \<\va_1, \ldots, \va_{k-1}, \va_k\>$ with $\lambda_{k}(\va_1, \ldots, \va_{k-1}, \va_k) = S$ for every $S \in S'$. Obviously, this approach might produce more than one representative with the fingerprint $S$ at level $k$. In order to account for such multiplicities, we formally annotate the $\alpha_{k,S}$ with indices~$\ell$.

In the end, we can derive equisatisfiability of $\varphi$ and $\varphi' :=$
	\begin{align*}
		\forall \vx_1& \exists \vy_1 \ldots \forall \vx_n \exists \vy_n. \psi \\
		&\wedge \bigvee^{\left|\fP^{n-1}[\mDNF]\right|}_{j_1=1} \Bigl(\vy_1 \approx \vec{c}_{\<j_1\>}
			\wedge \bigvee^{\left|\fP^{n-2}[\mDNF]\right|}_{j_2=1} \Bigl(\vy_2 \approx \vec{c}_{\<j_1, j_2\>}  \\
		&\qquad\wedge\quad\ldots\quad \Bigl( \ldots
			\quad \bigvee^{|[\mDNF]|}_{j_{n}=1} \vy_{n} \approx \vec{c}_{\<j_1, \ldots, j_{n}\>} \Bigr)\ldots\Bigr)\Bigr) ~,
	\end{align*}
where $\vy_k \approx \vec{c}_{\<j_1, \ldots, j_k\>}$ stands for $\bigwedge_{i=1}^{|\vy_k|} y_{k,i} \approx c_{\<j_1, \ldots, j_k\>;i}$ and all the $c_{\<j_1, \ldots, j_k\>; i}$ are fresh constant symbols.
In accordance with the approach described above, we introduce a nested form of range-restricting constraint this time.

In order to compute the number of constant symbols in $\varphi'$, we first define the notation $\twoup{k}{m}$ inductively: $\twoup{0}{m} := m$ and $\twoup{k+1}{m} := 2^{\left(\twoup{k}{m}\right)}$.
The number of constants that occur in $\varphi'$ is $|\consts(\varphi)| + \sum_{k=0}^{n-1} \Bigl( \prod_{\ell = k}^{n-1} \twoup{\ell}{\mDNF} \Bigr) \cdot |\vy_{n - k}| \leq \len(\varphi) + n \cdot \len(\varphi) \cdot \bigl( \twoup{n}{\len(\varphi)} \bigr)^n$. 
\begin{theorem}\label{theorem:SFGeneralComplexity}
	Let $\varphi := \exists\vz\, \forall \vx_1 \exists \vy_1 \ldots \forall \vx_n \exists \vy_n. \psi$ be a sentence in SF for some quantifier-free $\psi$ without non-constant function symbols. Satisfiability of $\varphi$ can be decided nondeterministically in time that is at most $n$-fold exponential in $\len(\varphi)$.
\end{theorem}

\subsubsection{Open formulas and dependencies}\label{section:SkolemizationOpenFormulas}

Let $\psi(\vx, \vy, \vz)$ be a quantifier-free first-order formula in which $\vx$ and $\vy$ are separated. 
We proceed analogously to the case of closed formulas treated in Section~\ref{section:SkolemizationSimpleSentence}.
The only difference is that the employed set constructions and the underlying fingerprint function depend on the valuation of the parameter tuple $\vz$.

We transform $\psi(\vx, \vy, \vz)$ into an equivalent formula $\bigvee_{k=1}^\mDNF \chi_k(\vx, \vz) \wedge \eta_k(\vy, \vz)$ such that $\chi_k$ and $\eta_k$ denote conjunctions of literals and $\mDNF$ is some non-negative integer. 
Given a structure $\A$ and a tuple $\vc \in \fU_\A^{|\vz|}$, we define the fingerprint function $\lambda_{\vc}$ by setting $\lambda_{\vc}(\vb) := \bigl\{ k \in [\mDNF] \bigm| \A,[\vx \Mapsto \vb, \vz \Mapsto \vc\,] \models \chi_k(\vx,\vz) \bigr\}$ for every $\vb \in \fU_\A^{|\vx|}$.
In order to complement $\lambda_{\vc}$, we construct the sets $\widehat{\fU}_{\vc, 1}, \ldots, \widehat{\fU}_{\vc, \mDNF} \subseteq \fU_\A^{|\vy|}$ such that $\widehat{\fU}_{\vc, k} := \bigl\{ \va \in \fU_\A^{|\vy|} \bigm| \A, [\vy \Mapsto \va, \vz \Mapsto \vc\,] \models \eta_k \bigr\}$. For every nonempty $\widehat{\fU}_{\vc, k}$ we pick one representative $\alpha_{\vc, k} \in \widehat{\fU}_{\vc, k}$.
Based on these sets, we define the mapping $\mu_{\vc} : \fU_\A^{|\vx|} \to [\mDNF]$ such that $\mu_{\vc}(\vb)$ yields some fixed index $k \in \lambda_{\vc}(\vb)$ for which $\fU_{\vc,k}$ is nonempty. If no such $k$ exists, $\mu_{\vc}(\vb)$ remains undefined. 

If $\A, [\vz \Mapsto \vc\,] \models \forall\vx \exists\vy. \psi(\vx, \vy, \vz)$, then for every $\vb \in \fU_\A^{|\vx|}$ we have $\A, [\vx \Mapsto \vb, \vy \Mapsto \alpha_{\vc, \mu_{\vc}(\vb)}, \vz \Mapsto \vc\,] \models \psi(\vx, \vy, \vz)$.
Hence, for every quantifier prefix $\Q\vz$---some of the variables in $\vz$ may be existentially quantified  and some universally---the following sentences are mutually equisatisfiable:
\begin{enumerate}[label=(\arabic{*}), ref=(\arabic{*})]
	\item $\Q\vz\, \forall\vx \exists\vy. \psi(\vx, \vy, \vz)$,
	\item $\Q\vz\, \forall\vx \exists\vy. \psi(\vx, \vy, \vz) \wedge \bigvee^{\mDNF}_{\ell=1}  \bigwedge_{i=1}^{|\vy|} y_i \approx g_{\ell,i}(\vz)$,
	\item $\Q\vz\, \forall\vx. \psi(\vx, \vy, \vz)\subst{y_1, \ldots, y_{|\vy|}}{f_1(\vx, \vz), \ldots, f_{|\vy|}(\vx, \vz)}$\\ 
		$\wedge \bigvee^{\mDNF}_{\ell=1} \bigwedge_{i=1}^{|\vy|} f_i(\vx, \vz) \approx g_{\ell,i}(\vz)$,
\end{enumerate}
where the $g_{\ell,i}$ are fresh function symbols of appropriate arity and the $f_i$ are Skolem functions of appropriate arity.
	
\medskip
Dually, we can transform $\psi(\vx, \vy, \vz)$ into an equivalent formula $\bigwedge^{\mCNF}_{i=1}$ $\chi'_i(\vx, \vz) \vee \eta'_i(\vy, \vz)$ in which the $\chi'_i$ and $\eta'_i$ are disjunctions of literals.
Given a structure $\A$ and a fixed tuple $\vc \in \fU_\A^{|\vz|}$ of domain elements, we define the fingerprint function $\lambda_{\vc} : \fU_\A^{|\vy|} \to \fP[\mCNF]$ such that for every $\va$ it holds $\lambda_{\vc}(\va) := \bigl\{ i \in [\mCNF] \bigm| \A,[\vy\Mapsto \va, \vz \Mapsto \vc\,] \models \eta'_i(\vy, \vz)\bigr\}$.
The function $\lambda_{\vc}$ induces a partition of $\fU_\A^{|\vy|}$ into $2^\mCNF$ equivalence classes. For every fingerprint $S \subseteq [\mCNF]$ such that there is some tuple $\va$ with $\lambda_{\vc}(\va) = S$ we fix a representative $\alpha_{\vc, S} \in \fU_\A^{|\vy|}$ of $\va$'s equivalence class, i.e.\ $\lambda_{\vc}(\alpha_{\vc, \lambda_{\vc}(\va)}) = \lambda_{\vc}(\va)$.

Then for all tuples $\vb \in \fU_\A^{|\vx|}$ and $\va \in \fU_\A^{|\vy|}$, we have 
$\A, [\vx \Mapsto \vb, \vy \Mapsto \va, \vz \Mapsto \vc\,] \models \psi(\vx, \vy, \vz)$
if and only if
$\A, [\vx \Mapsto \vb, \vy \Mapsto \alpha_{\vc,\lambda_{\vc}(\va)}, \vz \Mapsto \vc\,] \models \psi(\vx, \vy, \vz)$.
Consequently, for every quantifier prefix $\Q\vz$ the following sentences are mutually equisatisfiable:
\begin{enumerate}[label=(\arabic{*}), ref=(\arabic{*})]
	\item $\Q\vz\, \forall\vx \exists\vy. \psi(\vx, \vy, \vz)$,
	\item $\Q\vz\, \forall\vx \exists\vy. \psi(\vx, \vy, \vz) \wedge \bigvee^{2^\mCNF}_{j=1} \bigwedge_{i=1}^{|\vy|} y_i \approx g_{j,i}(\vz)$,
	\item $\Q\vz\, \forall\vx. \psi(\vx, \vy, \vz)\subst{y_1, \ldots, y_{|\vy|}}{f_1(\vx, \vz), \ldots, f_{|\vy|}(\vx, \vz)}$\\ 
		$\wedge \bigvee^{2^\mCNF}_{j=1} \bigwedge_{i=1}^{|\vy|} f_i(\vx, \vz) \approx g_{j,i}(\vz)$,
\end{enumerate}
where 
	the $g_{j,i}$ are fresh function symbols of appropriate arity and the $f_i$ are Skolem functions of appropriate arity.

Together with the dual case, we have inferred equisatisfiability of the sentences
$\Q\vz \forall\vx \exists\vy.$ $\psi(\vx, \vy, \vz)$,
and
$\Q\vz \forall\vx \exists\vy. \psi(\vx, \vy, \vz)$ $\wedge \bigvee^{m_*}_{j=1} \bigwedge_{i=1}^{|\vy|} y_i \approx g_{j,i}(\vz)$,
where $m_* := \min(2^\mCNF, \mDNF)$ and $\Q\vz$ can be any quantifier prefix closing the formula---some of the variables in $\vz$ may be existentially quantified  and some universally.

\subsubsection*{Relation to Henkin quantifiers}

In \cite{Henkin1961} Henkin introduced a generalized notion of quantifiers, sometimes called \emph{finite partially ordered quantifiers} or \emph{branching quantifiers} or \emph{nonlinear quantifiers}.
We have seen that separation of variables leads to a weaker dependency of  existentially quantified variables on universally quantified ones (cf.\ Proposition~\ref{proposition:TransposingQuantifierBlocks} and Lemma~\ref{lemma:SkolemizationOneQuantifierBlock}). The arity of Skolem functions may even be decreased at the price of a possibly exponential increase in the length of the formula (cf.\ Proposition~\ref{proposition:Blowup}). However, we have also seen that this does not lead to complete independence, as it would in the case of quantification in Henkin's style.

\begin{example}\label{example:DependencyAndHenkin}
	Consider the following equivalent sentences
	\begin{itemize}
		\item $\forall z \forall x \exists y.\; Q(z,y) \leftrightarrow P(x)$,
		\item $\forall z \forall x \exists y. \bigl(\neg Q(z,y) \wedge \neg P(x)\bigr) \vee \bigl(Q(z,y) \wedge P(x)\bigr)$,
		\item $\forall z \exists y_1 y_2 \forall x. \bigl( Q(z,y_1) \rightarrow P(x)\bigr) \wedge \bigl( P(x) \rightarrow Q(z,y_2)\bigr)$,
	\end{itemize}
	which we shall address by $\varphi$, $\varphi_\text{DNF}$ and $\varphi'$, respectively.

	\pagebreak[1]
	Standard Skolemization of $\varphi$ introduces a binary Skolem function $f_y$ and replaces the single occurrence of $y$ with the term $f_y(z,x)$, whose value fully depends on the value assigned to $x$ (distinct values for $x$ may cause distinct values for $f_y(z,x)$). In contrast to this, Skolemization of $\varphi'$ replaces the two variables $y_1, y_2$ with the terms $f_{y_1}(z)$ and $f_{y_2}(z)$, respectively, for two unary Skolem functions $f_{y_1}, f_{y_2}$. Interestingly, neither of the new terms depends on $x$.
	
	On the other hand, range-restricted Skolemization reconciles the two previous approaches by introducing three terms $f_y(z,x)$ and $g_1(z)$, $g_2(z)$---one depending on $x$ and two which are independent of $x$---and by adding the restriction $\forall z \forall x. f_y(z,x) \approx g_1(z) \vee f_y(z,x) \approx g_2(z)$. This limits the dependence of the value of $f_y(z,x)$ on the value assigned to $x$ to a finite degree (over infinite domains, distinct values for $x$ cannot always result in distinct values for $f_y(z,x)$). However, it still does not lead to complete independence from $x$.
	
	Henkin quantifiers can explicitly express dependencies of existentially quantified variables on universally quantified ones. For instance, we could write $\psi := \forall z \forall x \exists_z y. Q(z,y) \leftrightarrow P(x)$ to express that the value of $y$ depends on the value of $z$ but not on $x$'s value. Then $\psi$ is equisatisfiable to $\psi_{\text{Sk}} := \forall z \forall x. Q(z,f'_y(z)) \leftrightarrow P(x)$ for some Skolem function $f'_y$. Due to the enforced independence of $y$ from $x$ in $\psi$, it is easy to construct a model $\A$ of $\varphi$ which cannot be extended to a model $\B$ of $\psi_{\text{Sk}}$ (e.g., set $\fU_\A := \{0,1\}$, $P^\A := \{0\}$ and $Q^\A := \{\<0,0\>, \<1,1\>\}$). Finding a satisfying extension $\B$ of $\A$ is not a problem in any other case of Skolemization that we have described above (for example, set $\fU_\B := \fU_\A$, $P^\B := P^\A$, $Q^\B := Q^\A$, $f_{y_1}^\B(0) := g_1^\B(0) := 1$, $f_{y_1}^\B(1) := g_1^\B(1) := 0$, $f_{y_2}^\B(0) := g_2^\B(0) := 0$, $f_{y_2}^\B(1)  := g_2^\B(1) := 1$, and $f_y^\B(0,\fa) := g_2^\B(\fa)$, $f_y^\B(1,\fa) := g_1^\B(\fa)$ for every $\fa \in \{0,1\}$).
\end{example}
Altogether, the example illustrates that separation of existentially quantified variables from universally quantified ones does lead to a certain degree of independence, but it does not reach the level of independence Henkin quantifiers can guarantee. This is not at all surprising, because Henkin quantifiers increase the expressiveness of first-order logic.

\subsection{Extensions of the Separated Fragment}\label{section:FurtherExtensionsSeparated}

In this section we describe methods extending SF into a proper superset of the Full Monadic Fragment (with unary function symbols but without equality)---shown to be decidable independently by \citet{Lob1967} and \citet{Gurevich1971}---and the Relational Monadic Fragment with equality (but without non-constant function symbols). We will show how sentences from both fragments can be transformed into ones pursuant to Definition~\ref{definition:SeparatedFragment} under an at most quadratic increase in the length of the formulas.

Adopting a method already used by L\"ob  in \cite{Lob1967} and also by Gr\"adel (cf.\ proof of Proposition 6.2.7 in \cite{Borger1997}), we can handle unary function symbols under certain restrictions.
\begin{proposition}\label{proposition:UnaryFunctions}
	Let $\varphi$ be a first-order sentence without non-unary function symbols (constants are allowed).
	If the unary function symbols exclusively occur in atoms starting with a unary predicate symbol, then we can find an equisatisfiable sentence $\varphi'$ without non-constant function symbols such that any model $\B$ of $\varphi'$ can be transformed into a model $\A$ of $\varphi$ over the same domain. 
	The length of $\varphi'$ lies in $\O(\len(\varphi))$.
	Moreover, if $\varphi$ belongs to SF, then $\varphi'$ is also an SF sentence.
\end{proposition}
\begin{proof}
	Let $f_1, \ldots, f_k$ be the unary function symbols occurring in $\varphi$. We apply the following transformation iteratively.
	Assume $\varphi$ contains the atom $P(f_i(t))$ for some term $t$. We may transform $\varphi$ into 
		$\varphi\subst{P(f_i(t)) \bigm/ R(t)} \wedge \forall x.P(f_i(x)) \leftrightarrow R(x)$,
	where the $R$ is a fresh unary predicate symbol and $\varphi\subst{P(f_i(t)) \bigm/ R(t)}$ is the formula we obtain from $\varphi$ by replacing every occurrence of $P(f_i(t))$ by $R(t)$.	
	Exhaustive application of this transformation to $\varphi$ yields a sentence $\varphi''$ of the form $\psi \wedge \bigwedge_{i=1}^k \bigwedge_j \forall x. P_j(f_i(x)) \leftrightarrow R_{i,j}(x)$, where $\psi$ does not contain any of the $f_i$ anymore. If we conceive the $f_i$ in $\varphi'$ as Skolem functions and revert the Skolemization, the $f_i$ vanish completely and we end up with the equisatisfiable sentence $\varphi' := \psi \wedge \forall x \exists y_1\ldots y_k.  \bigwedge_{i=1}^k \bigwedge_j P_j(y_i) \leftrightarrow R_{i,j}(x)$.
	
Because of $\len(\psi) \leq \len(\varphi)$ and since for any occurrence of an $f_i$ in $\varphi$ at most one new conjunct of a fixed length is introduced, it holds $\len(\varphi') \in \O(\len(\varphi))$.
\end{proof}

If we now allow for unary function symbols occurring in monadic atoms in SF, as Proposition~\ref{proposition:UnaryFunctions} suggests, this extended fragment becomes a proper superset of the Full Monadic fragment without equality.

Equality in SF, as defined in Definition~\ref{definition:SeparatedFragment}, is subject to the separation condition. However, there is no such restriction in monadic formulas with equality. For instance, while the sentence $\forall x \exists y. x \approx y$ is admitted for the Relational Monadic Fragment with equality, the sets $\{x\}$ and $\{y\}$ are not separated in this sentence.
Next, we show why the separation restriction may be discarded for equations, if the sentence at hand exhibits the small model property.
We start by treating the case of monadic sentences with equality but without non-constant function symbols.

\begin{proposition}\label{proposition:MonadicEqualityEncoding}
	For every sentence $\varphi$ in the Relational Monadic Fragment with equality we can construct an equisatisfiable relational monadic sentence  $\varphi'$ without equality. Moreover, the length of $\varphi'$ is of order $\O(\len(\varphi)^2)$.
\end{proposition}
\begin{proof}
		Let $\varphi$ be a first-order sentence without non-constant function symbols and containing at most the unary predicate symbols $P_1, \ldots, P_n$ (besides the equality predicate). 
	Let $k$ be the number of occurrences of quantifiers in $\varphi$ plus the number of constant symbols in $\varphi$, and set $\kappa := \lceil \log_2 k \rceil$.
	We extend the underlying signature with fresh unary predicate symbols $Q_1, \ldots, Q_{\kappa}$ and  define the formula 
		\[ \psi_{\approx}(x,y) := \bigwedge_{i=1}^n \bigl( P_i(x) \leftrightarrow P_i(y) \bigr) \wedge  \bigwedge_{i=1}^\kappa \bigl( Q_i(x) \leftrightarrow Q_i(y) \bigr) ~.\]
	The length of $\psi_{\approx}(x,y)$ lies in $\O(\len(\varphi))$.	
	Let $\varphi'$ be constructed from $\varphi$ by replacing every occurrence of an equation $s \approx t$ with $\psi_{\approx}(x,y)\bigl[x/s,$ $y/t\bigr]$. 
	Since $s,t$ cannot contain non-constant function symbols, the length of $\psi_{\approx}(x,y)\subst{x/s,y/t}$ is the same as that of $\psi_{\approx}(x,y)$. 
	We show that $\varphi$ is satisfiable if and only if $\varphi'$ is so.
	
	Let $\A$ be a model of $\varphi$ over the domain $\fU_\A$. We construct a fingerprint function $\lambda : \fU_\A \to \fP[n]$ for which $\lambda(\fa) := \bigl\{ i \in [n] \bigm| \fa\in P_i^\A \bigr\}$ for every $\fa \in \fU_\A$. We now partition $\fU_\A$ into parts $\fU_{S}$ such that every $\fU_{S}$ contains exactly the elements $\fa$ with the fingerprint $\lambda(\fa) = S$.
	
	Starting from $\A$ we construct an interpretation $\A'$ by arbitrarily choosing subsets $\fU'_{S} \subseteq \fU_{S}$ for every $S\subseteq [n]$ such that $|\fU'_{S}| = \min(|\fU_{S}|, k)$, and by defining $\A'$'s domain to be $\fU_{\A'} := \bigcup_{S \subseteq [n]} \fU'_{S}$. Moreover, we set $P_i^{\A'} := P_i^A \cap \fU_{\A'}$ for all $P_i$ and define the predicates $Q_i^{\A'}$ so that for every $S \in [n]$ and for all elements $\fa, \fa' \in \fU'_{S}$ we can find an index $j$ such that $\fa \in Q_j^{\A'}$ and $\fa' \not\in Q_j^{\A'}$ or vice versa.
	Clearly, for every variable assignment $\beta : X \to \fU_{\A'}$ it then holds $\A',\beta \models x \approx y$ if and only if $\A',\beta \models \psi_{\approx}(x,y)$
	
	Since a formula with $k$ different variables and constant symbols cannot distinguish more than $k$ elements of the same color, $\A'$ must also be a model of $\varphi$. This together with the equivalence of $x \approx y$ and $\psi_{\approx}(x,y)$ under $\A'$ entails that $\A'$ is a model of $\varphi'$.
	This proves the ``only if''-direction.
	
	Now let $\B'$ be an arbitrary model of $\varphi'$. We now construct a fingerprint function $\lambda' : \fU_{\B'} \to \fP[n+\kappa]$ analogously to $\lambda$, but now also taking the predicates $Q_i^{\B'}$ into account. Again, we partition $\fU_{\B'}$ into parts $\fU_{S}$ so that all elements in a set $\fU_{S}$ are assigned the same color by $\lambda'$. We now define the new universe $\fU_{\B}$ by arbitrarily picking exactly one element from every part $\fU_{S}$, $S \subseteq [n+\kappa]$. In addition, we set $P_i^\B := P_i^{\B'} \cap \fU_\B$ and $Q_i^\B := Q_i^{\B'} \cap \fU_\B$.
	
	Since $\varphi'$ can only distinguish elements of $\fU_\B$ by means of their belonging to predicates $P_i^\B$ and $Q_i^\B$, it is clear that $\B$ is a model of $\varphi'$.
	Moreover, for any variable assignment $\beta : X \to \fU_{\B}$ we get $\B,\beta \models \psi_{\approx}(x,y)$ if and only if $\B,\beta \models x\approx y$.
	Put together, these two facts entail $\B \models \varphi$.	
\end{proof}
More generally, we can formulate the above result for all sentences which exhibit the small model property.
\begin{proposition}\label{proposition:EqualityEncodingOverFiniteDomains}
	Let $\varphi$  be a first-order sentence with equality.
	Suppose we can compute a positive integer $k$ such that if $\varphi$ is satisfiable, then there is a model $\A \models \varphi$ over a universe of cardinality at most $k$.
	Let $\kappa := \lceil \log_2 k \rceil$. We can transform $\varphi$ into an equisatisfiable sentence $\varphi'$ without equality using only the vocabulary of $\varphi$ plus $\kappa$ fresh unary predicate symbols $Q_1, \ldots, Q_\kappa$. 	
	The length of $\varphi'$ lies in $\O(\kappa\cdot \len(\varphi)^3)$.
	Moreover, if existentially quantified variables are separated from universally quantified ones in all atoms in $\varphi$ except for equations, then they are separated in \emph{all} atoms in $\varphi'$.
\end{proposition}
\begin{proof}
	We generalize the argument from the proof of Proposition~\ref{proposition:MonadicEqualityEncoding}. 
	Since $\varphi$ may also contain non-unary predicate symbols and non-constant function symbols, the encoding of equality requires additional congruence axioms.
	First, $\psi_{\approx} (x,y)$ becomes slightly simpler:
		\[ \psi_{\approx}(x,y) := \bigwedge_{i=1}^{\kappa} \bigl( Q_i(x) \leftrightarrow Q_i(y) \bigr) ~.\]
	In addition, we need the following congruence axioms:
		\[ \psi_{\text{pred}} := \forall x y. \psi_{\approx}(x,y) \;\;\longrightarrow\;\; \bigwedge_P \bigwedge_{i = 0}^{\text{arity}(P)-1} \bigl( \forall \vu_i \vu'_i.\; P(\vu_i, x,  \vu'_i) \to P( \vu_i, y,  \vu'_i) \bigr) ~,\]
	where $P$ ranges over all predicate symbols in $\varphi$, except for the equality predicate, and $\vu_i,  \vu'_i$ are disjoint sets of variables of cardinality $|\vu_i| = i$ and $|\vu'_i| = \text{arity}(P) - i - 1$ for the respective $P$,
	and
		\[ \psi_{\text{func}} := \forall x y. \psi_{\approx}(x,y) \;\;\longrightarrow\;\; \bigwedge_f \bigwedge_{i = 0}^{\text{arity}(f)-1} \bigl( \forall \vv_i \vv'_i.\; \psi_{\approx}(x,y)\subst{x/f(\vv_i, x,  \vv'_i), y/f( \vv_i, y,  \vv'_i)} \bigr) ~, \]
	where $f$ ranges over all non-constant function symbols in $\varphi$, and $\vv_i,  \vv'_i$ are disjoint sets of variables of cardinality $|\vv_i| = i$ and $|\vv'_i| = \text{arity}(f) - i - 1$ for the respective $f$.

	We construct $\varphi'$ from $\varphi$ by replacing every occurrence of an equation $s \approx t$ with the formula $\psi_{\approx}(x,y)\subst{x/s,y/t}$, and by conjunctively connecting $\psi_{\text{pred}}$ and $\psi_{\text{func}}$ to it. 
	
	\medskip	
	Let $\A$ be a model of $\varphi$ over the finite domain $\fU_\A := \{\fa_1, \ldots, \fa_k\}$. 	
	Starting from $\A$, we construct a structure $\B$ such that
	$\fU_\B := \fU_\A$,
	$P^\B := P^\A$ for every predicate symbol $P$ occurring in $\varphi$,
	$c^\B := c^\A$ for every constant symbol $c$ occurring in $\varphi$,
	$f^\B := f^\A$ for every function symbol $f$ occurring in $\varphi$.
	Moreover, we define $Q_1^\B, \ldots, Q_\kappa^\B$ such that for all distinct $\fa_{i_1}, \fa_{i_2} \in \fU_\B$ we find some $j$, $1\leq j\leq \kappa$, for which $\fa_{i_1} \in Q_{j}^\B$ if and only if $\fa_{i_2} \not\in Q_{j}^\B$. 
	This is possible, because, having $\kappa = \lceil \log_2 k \rceil$ unary predicates, one can distinguish $2^\kappa \geq k$ domain elements.
	Hence, for any variable assignment $\beta$ it holds $\B,\beta \models \psi_{\approx}(x,y)$ if and only if $\beta(x) = \beta(y)$ if and only if $\B,\beta \models x \approx y$.
	Consequently, $\B$ is a model of $\varphi$ as well as of $\varphi'$.
	The finishes the ``only if''-part of the proof.
	
	\medskip	
	Now let $\B$ be a model of $\varphi'$.
	Hence, $\B$ is also a model of $\psi_{\text{pred}}$ and $\psi_{\text{func}}$.
	We define a fingerprint function $\lambda : \fU_\B \to \fP[\kappa]$ where we set $\lambda(\fa) := \bigl\{ i\in [\kappa] \bigm| \fa \in Q_i^\B \bigr\}$ for every $\fa \in \fU_\B$. 
	Using $\lambda$, we decompose $\fU_{\B}$ into parts $\fU_{S} := \{ \fa\in \fU_\B \mid \lambda(\fa) = S \subseteq [\kappa] \}$ and fix for every such set one representative $\alpha_{S} \in \fU_{S}$. 
	
	We now define a new structure $\A$ based on the representatives $\alpha_{S}$.
	To this end, we set 
	$\fU_{\A} := \{ \alpha_{\lambda(\fa)} \mid \fa \in \fU_\B \}$,
	$P^{\A} := P^\B \cap \fU_{\A}^m$ for every $m$-ary predicate symbol $P$ occurring in $\varphi$,
	$c^{\A} := \alpha_{\lambda(c^\B)}$ for every constant symbol $c$ occurring in $\varphi$,
	$f^{\A}(\va) := \alpha_{\lambda(f^\B(\va))}$ for every $m$-ary function symbol $f$ occurring in $\varphi$ and every tuple $\va \in \fU_{\A}^m$.
	
	We make the following observations:
	\begin{enumerate}[label=(\arabic{*}), ref=\arabic{*}]
		\item\label{enum:proofEqualityEncoding:I} $\lambda(\alpha_{\lambda(\fa)}) = \lambda(\fa)$ for every $\fa \in \fU_\B$.
		\item\label{enum:proofEqualityEncoding:II} $\alpha_{\lambda_B(\fa)} = \fa$ for every $\fa \in \fU_\A$.
		\item\label{enum:proofEqualityEncoding:III} Given $\fa, \fb \in \fU_\B$, $\lambda(\fa) = \lambda(\fb)$ implies $\B, \beta \models \psi_{\approx}(x,y)$ for every variable assignment $\beta$ that fulfills $\beta(x) = \fa$ and $\beta(y) = \fb$.
		\item\label{enum:proofEqualityEncoding:IV} Given $\fa_1, \ldots, \fa_m, \fb_1, \ldots, \fb_m \in \fU_\B$ such that $\lambda_B(\fa_i) = \lambda(\fb_i)$ for $i = 1, \ldots, m$, it holds
			\begin{enumerate}[label=(\alph{*}), ref=\alph{*}]
				\item\label{enum:proofEqualityEncoding:IV:I} $\lambda(f^\B(\fa_1, \ldots, \fa_m)) = \lambda(f^\B(\fb_1, \ldots, \fb_m))$ and
				\item\label{enum:proofEqualityEncoding:IV:II} $\<\fa_1, \ldots, \fa_m\> \in P^\B$ if and only if $\<\fb_1, \ldots, \fb_m\> \in P^\B$
			\end{enumerate}
			for every $m$-ary function symbol $f$ and every $m$-ary predicate symbol $P$ occurring in $\varphi'$, respectively.
	\end{enumerate}
	All of these observations are consequences of the definition of the coloring $\lambda$ and our assumption $\B \models \psi_{\text{pred}} \wedge \psi_{\text{func}}$.
	
	\begin{description}
		\item\underline{Claim I:}
			Let $t$ be some term occurring in $\varphi'$, and let $\beta, \beta'$ be two variable assignments such that $\beta'(x) = \alpha_{\lambda(\beta(x))}$ holds for every variable $x$ occurring in $t$.
			We observe $\lambda\bigl( \A(\beta')(t)) \bigr) = \lambda\bigl( \B(\beta)(t)) \bigr)$.
			
		\item\underline{Proof:}
			The proof proceeds by induction on the structure of $t$.
			\begin{description}
				\item Base cases: 
					Suppose $t = c$ for some constant symbol $c$. 
					By definition of $\A$, it holds $c^\A = \alpha_{\lambda(c^\B)}$. 
					Observation~(\ref{enum:proofEqualityEncoding:I}) thus entails $\lambda\bigl( \A(\beta')(c) \bigr) = \lambda\bigl( c^\A \bigr) = \lambda\bigl( \alpha_{\lambda(c^\B)} \bigr) = \lambda\bigl( c^\B \bigr) = \lambda\bigl( \B(\beta)(c) \bigr)$.
					
					Suppose $t = x$ for some variable $x$. By assumption, we have $\beta'(x) = \alpha_{\lambda(\beta(x))}$.
					Observation~(\ref{enum:proofEqualityEncoding:I}) thus entails $\lambda\bigl( \A(\beta')(x) \bigr) = \lambda\bigl( \beta'(x) \bigr) = \lambda\bigl( \alpha_{\lambda(\beta(x))} \bigr) = \lambda\bigl( \beta(x) \bigr) = \lambda\bigl( \B(\beta)(x) \bigr)$.
					
				\item Inductive case:
					Let $t = f(s_1, \ldots, s_m)$ for some $m$-ary function symbol $f$ and arbitrary terms $s_1, \ldots, s_m$.
					By definition of $f^\A$ and Observation~(\ref{enum:proofEqualityEncoding:I}), we get
					\begin{align*}
						&\lambda\bigl( \A(\beta')(f(s_1, \ldots, s_m)) \bigr) \\
						&= \lambda\bigl( f^\A(\A(\beta')(s_1), \ldots, \A(\beta')(s_m)) \bigr) \\
						&= \lambda\bigl( \alpha_{\lambda(f^\B(\A(\beta')(s_1), \ldots, \A(\beta')(s_m)))} \bigr) \\
						&= \lambda\bigl( f^\B(\A(\beta')(s_1), \ldots, \A(\beta')(s_m)) \bigr) ~.
					\end{align*}	
					By induction, we conclude $\lambda\bigl( \A(\beta')(s_i) \bigr)  =  \lambda\bigl( \B(\beta)(s_i) \bigr)$ for $i = 1, \ldots, m$.
					Hence, Observation~(\ref{enum:proofEqualityEncoding:IV}\ref{enum:proofEqualityEncoding:IV:I}) leads to
					\begin{align*}
						&\lambda\bigl( f^\B(\A(\beta')(s_1), \ldots, \A(\beta')(s_m)) \bigr) \\
						&= \lambda\bigl( f^\B(\B(\beta)(s_1), \ldots, \B(\beta)(s_m)) \bigr) \\ 
						&= \lambda\bigl( \B(\beta)(f(s_1, \ldots, s_m)) \bigr) ~.
					\end{align*}
					Consequently, $\lambda\bigl( \A(\beta')(f(s_1, \ldots, s_m)) \bigr)$ equals $\lambda\bigl( \B(\beta)(f(s_1, \ldots, s_m)) \bigr)$.
					\strut\hfill$\Diamond$
			\end{description}
		
		\item\underline{Claim II:}
			Let $A$ be some atom occurring in $\varphi'$, and let $\beta, \beta'$ be two variable assignments such that $\beta'(x) = \alpha_{\lambda(\beta(x))}$ holds for every variable $x$ occurring in $A$.
			$\B,\beta \models A$ holds if and only if $\A,\beta'\models A$.
			
		\item\underline{Proof:}		
			Since $\varphi'$ does not contain equality, $A$ must be of the shape $P(t_1, \ldots, t_m)$ for some $m$-ary predicate symbol $P$ and terms $t_1, \ldots, t_m$.
			$\B,\beta \models P(t_1, \ldots, t_m)$ holds if and only if $\bigl\< \B(\beta)(t_1), \ldots, \B(\beta)(t_m) \bigr\> \in P^\B$.
			By Claim I, we know $\lambda\bigl( \B(\beta)(t_i) \bigr) = \lambda\bigl( \A(\beta')(t_i) \bigr)$ for $i = 1, \ldots, m$, and thus Observation~(\ref{enum:proofEqualityEncoding:IV}\ref{enum:proofEqualityEncoding:IV:II}) entails that $\bigl\< \B(\beta)(t_1), \ldots, \B(\beta)(t_m) \bigr\> \in P^\B$ holds true if and only if $\bigl\< \A(\beta')(t_1), \ldots, \A(\beta')(t_m) \bigr\> \in P^\B$ does.
			And since all the domain elements $\A(\beta')(t_i)$ belong to $\fU_\A$, it follows that $\bigl\< \A(\beta')(t_1), \ldots, \A(\beta')(t_m) \bigr\> \in P^\B$ is equivalent to $\bigl\< \A(\beta')(t_1), \ldots, \A(\beta')(t_m) \bigr\> \in P^\B \cap \fU_\A^m = P^\A$.
			Consequently, we have $\A,\beta' \models P(t_1, \ldots, t_m)$ if and only if $\B,\beta \models P(t_1, \ldots, t_m)$ holds.
			\strut\hfill$\Diamond$
		
		\item\underline{Claim III:}
			$\B \models \varphi'$ implies $\A \models \varphi'$.			
			
		\item\underline{Proof:}			
			The proof proceeds by induction on the structure of $\varphi'$.
			However, we need to slightly generalize the claim first:
				given any two variable assignments $\beta, \beta'$ for which $\beta'(x) = \alpha_{\lambda(\beta(x))}$ for every variable $x$ occurring freely in $\varphi'$, we observe that $\B,\beta \models \varphi'$ entails $\A,\beta' \models \varphi'$.
			
			The base case is already handled by Claim II.
			Of the inductive cases only the ones with quantifiers are of interest, the others are trivial.
			
			Suppose $\varphi'$ is of the form $\forall x.\psi$.
				In order to show $\A,\beta' \models \forall x.\psi$, we must show that $\A,\beta'[x\Mapsto \fa] \models \psi$ holds for every $\fa \in \fU_\A$.
				By Observation~\ref{enum:proofEqualityEncoding:II} and our assumption about $\beta$ and $\beta'$, we have $\bigl(\beta'[x\Mapsto \fa]\bigr)(v) = \alpha_{\lambda((\beta[x\Mapsto \fa])(v))}$ for every variable $v$ occurring freely in $\psi$, and thus the updated variable assignments $\beta[x\Mapsto \fa]$ and $\beta'[x\Mapsto \fa]$ fulfill the requirements of the inductive hypothesis.
				By induction, we derive $\A,\beta'[x\Mapsto \fa] \models \psi$ for any $\fa \in \fU_\A$ from $\B,\beta[x \Mapsto \fa] \models \psi$.
			
			Suppose $\varphi'$ is of the form $\exists y.\psi$.
				We need to find some $\fa \in \fU_\A$ for which $\A,\beta'[y\Mapsto \fa] \models \psi$ holds true.
				Since we assume $\B,\beta \models \exists y.\psi$, there is some $\fb \in \fU_\B$ such that $\B,\beta[y\Mapsto \fb] \models \psi$.
				We set $\fa := \alpha_{\lambda(\fb)}$.
				Hence, we have $\bigl(\beta'[y\Mapsto \fa]\bigr)(v) = \alpha_{\lambda((\beta[y\Mapsto \fb])(v))}$ for every variable $v$ occurring freely in $\psi$.			
				By induction, we obtain $\A,\beta'[y\Mapsto \fa] \models \psi$ from $\B,\beta[y \Mapsto \fb] \models \psi$.
			
			\medskip
			This finishes the proof of $\B,\beta \models \varphi'$ implying $\A,\beta' \models \varphi'$.
			Since $\varphi'$ is assumed to be a sentence and thus does not contain free occurrences of any variable, it immediately follows that $\B \models \varphi'$ entails $\A \models \varphi'$			
			\strut\hfill$\Diamond$			
	\end{description}
	
	Due to the fact that all domain elements in $\fU_\A$ are assigned a different fingerprint by $\lambda$, we may conclude for every $\beta$ that $\A,\beta \models \psi_{\approx}(x,y)$ if and only if $\A,\beta \models x \approx y$.
	Consequently, $\A \models \varphi'$ entails $\A \models \varphi$.
	
	\medskip
	An inspection of the described construction leads to the following upper bounds on the length of the resulting (sub)formulas:
		$\len\bigl(\psi_{\approx} (x,y)\bigr) \in \O(\kappa)$,
		$\len\bigl(\psi_{\approx}(x,y)\subst{x/s,y/t}\bigr)\in \O(\kappa\cdot \len(\varphi))$,
		$\len\bigl(\psi_{\text{pred}}\bigr)\in \O(\kappa + \len(\varphi)^3)$, and 
		$\len\bigl(\psi_{\text{func}}\bigr)\in \O(\kappa \cdot \len(\varphi)^3)$.
	All in all, this yields $\len(\varphi') \in \O(\kappa\cdot \len(\varphi)^3)$.
\end{proof}

The just proven proposition means that some sentences, which are \emph{almost} SF sentences, can be transformed into equisatisfiable SF sentences, if they exhibit the small model property.
More precisely, being \emph{almost} SF in this context means they do not contain non-constant function symbols and fulfill the separation condition only for non-equational atoms.

In fact, we can even allow certain occurrences of non-constant function symbols. For instance, the setting of unary function symbols described in Proposition~\ref{proposition:UnaryFunctions} is possible, i.e.\ equations of the form $f(g(f(x))) \approx g(h(y))$ with $x$ being universally quantified and $y$ existentially, may be allowed, as long as the sentence at hand has a finite model, for which we can computer an upper bound on its size. In this case, first applying the transformation given in the proof of Proposition~\ref{proposition:MonadicEqualityEncoding} and subsequently the construction from the proof of Proposition~\ref{proposition:UnaryFunctions} finally leads to an SF sentence. Combinations with other translation methods are also conceivable.

\begin{remark}
	Other methods for the elimination of the distinguished equality predicate have been proposed. 
	For instance, \citet{Dreben1979} (Chapter 8, Sections 1 and 2) describe \emph{negative identity reduction}, which introduces a binary predicate $I$ to capture equality. 
	However, this approach can obviously not resolve non-separated equations into separated atoms, and therefore does not exactly meet our particular needs. 
	On the other hand, this method is not restricted to sentences possessing the small model property.
	
	For an automated-reasoning perspective on the elimination of equality, consult \cite{Baumgartner2009, Bachmair1998, Brand1975, Letz2002}, for instance.
\end{remark}







\section{Separation and Automated Reasoning}\label{section:AutomatedReasoning}

The size of the search space of a first-order automated reasoning procedure is related to the size of the relevant subset of the
Herbrand base that is actually explored by the procedure. 
Automated reasoning on first-order formulas with an explicit or implicit finite Herbrand base 
has attracted a lot of attention in recent 
years~\cite{InstGen03,ModelEvolutionLemma06,Piskac2010,InfinoxRef11,HillenbrandWeidenbach13,Korovin13,SGGSExposition,Alagi2015}.
In particular, all these procedures are decision procedures for the BS(R) Fragment.
But also a clause set with an explicit finite domain axiom
or a clause set enjoying an explicit or implicit acyclic atom structure
generates only a finite relevant subset of the
overall Herbrand base. If the relevant subset of the Herbrand base is infinite, 
implying the presence of non-constant function
symbols, then the search space of an automated reasoning procedure becomes infinite and it does not terminate anymore, in general.
Even for a finite relevant subset generated over a finite number of constants the actual size of 
the set has an important influence on the explored search space.

In this section we apply our results to the benefit of the above mentioned automated reasoning procedures.
We show the techniques for one quantifier block alternation, but they can be generalized to several blocks
analogously to the results in Section~\ref{section:SeparatedFragment}. 
Note that given a set of sentences where not all sentences are separated, our results are
still applicable to the separated sentences.
Consider the separated sentence $\forall\vx \exists\vy. \psi(\vx, \vy)$. 
This formula is satisfiable if and only if it is satisfiable over a 
domain of $m_*|\vy| + |\consts(\psi)|$ (see Lemma~\ref{lemma:SkolemizationOneQuantifierBlock})
elements\newline
\centerline{$\forall\vx \exists\vy. \psi(\vx, \vy) \wedge \bigvee^{m_*}_{\ell=1} \bigwedge_{i=1}^{|\vy|} y_i \approx c_{\ell,i}$.}
Skolemizing the $\vy$ and explicitly representing the substitution in $\psi$ by equations yields\newline
\centerline{$\forall\vx\vy. \Bigl( \bigl(\bigwedge_{i=1}^{|\vy|} f_i(\vx) \approx y_i\bigr) \rightarrow \psi(\vx, \vy) \Bigr) \wedge \bigvee^{m_*}_{\ell=1}\bigwedge_{i=1}^{|\vy|} f_i(\vx) \approx c_{\ell,i}$}
that is a first-order sentence with an explicit finite domain axiom, so the relevant Herbrand base is finite.
Still, reasoning with equality on the $f_i$
is not needed here and we can further simplify the formula to the equisatisfiable sentence\newline
\centerline{$\forall\vx\vy. \Bigl( \bigl(\bigwedge_{i=1}^{|\vy|} R_i(\vx, y_i)\bigr) \rightarrow \psi(\vx, \vy) \Bigr) \wedge \bigvee^{m_*}_{\ell=1}\bigwedge_{i=1}^{|\vy|} R_i(\vx, c_{\ell,i})$}
where we replaced the $f_i$ by relations $R_i$ without the need to add further axioms. The transformation preservers
satisfiability. Totality of the $f_i$ after translation to $R_i$ is guaranteed by the finite
domain axiom $\bigvee^{m_*}_{\ell=1}\bigwedge_{i=1}^{|\vy|} f_i(\vx) \approx c_{\ell,i}$. If some $R_i$ interpretation
contains more than one value for an $\vx$ assignment, all values except one can simply be dropped, preserving
satisfiability. This is a consequence of the fact that all positive $f_i$ equational occurrences 
($R_i$ occurrences) are exactly in the finite domain axiom.
Eventually, if $\psi(\vx, \vy)$
does not contain any non-constant function symbols,
we moved a separated sentence that
would have resulted after CNF generation in an infinite Herbrand base to a clause set of the BS(R) Fragment.

The size of the finite domain axiom is worst-case exponential in the number $m_{\text{CNF}}$
of clauses generated out of $\forall\vx \exists\vy.\psi(\vx, \vy)$ by a 
CNF procedure without renaming~\cite{NonnengartWeidenbach01handbook}. Thus, if the number
of clauses can be reduced, it improves the size of the finite domain axiom. Redundant clauses do not contribute
to $m_{\text{CNF}}$. For example, if two clauses subsume each other, only one needs to be considered for $m_{\text{CNF}}$.
Note that $m_{\text{CNF}}$ without considering redundancy can be computed without actually generating the
CNF~\cite{NonnengartWeidenbach01handbook}. Renaming, i.e.~the replacement of subformulas via fresh predicates, 
cannot be applied for determining $m_{\text{CNF}}$, 
because it may generate atoms in which  variables $\vx$ and $\vy$ are no longer separated.
An analogous argument holds for the transformation into DNF.

The only formula whose length may grow exponentially in the length of $\psi(\vx, \vy)$
in the above transformation is the finite domain axiom. Instead of adding this axiom to the formula 
it could as well be built into a decision procedure for the BS(R) Fragment. Although this is subject to
future work, one line of extension for resolution refutation building procedures, for example,  is to start with a 
finite domain axiom for a small number of constants. In case of a refutation, extend the refutation derivation and the involved
finite domain axiom by literals with further constants, until the overall limit $m_*|\vy| + |\consts(\psi)|$ 
is reached or a model is found. In case the finite domain axiom is a priori small, an explicit instantiation of $\psi(\vx, \vy)$
can support the efficiency of an automated reasoning procedure. Explicit instantiation is the typical method used by SMT solvers
when confronted with quantification~\cite{GeMoura09}.
In general, an exponential number of domain elements can be necessary for finding a model of
a separated sentence. So the exponential growth cannot be escaped.

\begin{lemma}\label{lemma:RelaxEqualityInFiniteDomainAxiom}
	The following sentences are equisatisfiable.
	(We assume that the $f_i$ are fresh Skolem functions of appropriate arity, the $R_i$ are fresh predicate symbols, and the constant symbols $c_{\ell, i}$ do not occur in $\psi$.)
	\begin{enumerate}[label=(\arabic{*}), ref=(\arabic{*})]
		\item\label{enum:EquisatSingleBlock:II} $\forall\vx \exists\vy. \psi(\vx, \vy) \wedge \bigvee^{m_*}_{\ell=1} \bigwedge_{i=1}^{|\vy|} y_i \approx c_{\ell,i}$,
		\item\label{enum:EquisatSingleBlock:III} $\forall\vx. \psi(\vx, \vy)\subst{y_1/f_1(\vx), \ldots, y_{|\vy|}/f_{|\vy|}(\vx)} \wedge \bigvee^{m_*}_{\ell=1}  \bigwedge_{i=1}^{|\vy|} f_i(\vx) \approx c_{\ell,i}$,
		\item\label{enum:EquisatSingleBlock:IV} $\forall\vx \vy. \Bigl( \bigl(\bigwedge_{i=1}^{|\vy|} f_i(\vx) \approx y_i\bigr) \rightarrow \psi(\vx, \vy) \Bigr) \wedge \bigvee^{m_*}_{\ell=1}\bigwedge_{i=1}^{|\vy|} f_i(\vx) \approx c_{\ell,i}$,
		\item\label{enum:EquisatSingleBlock:V} $\forall\vx \vy. \Bigl( \bigl(\bigwedge_{i=1}^{|\vy|} R_i(\vx, y_i)\bigr) \rightarrow \psi(\vx, \vy) \Bigr) \wedge \bigvee^{m_*}_{\ell=1}\bigwedge_{i=1}^{|\vy|} R_i(\vx, c_{\ell,i})$.
	\end{enumerate}
\end{lemma}
\begin{proof}	

\ref{enum:EquisatSingleBlock:III} is the Skolemization of \ref{enum:EquisatSingleBlock:II}. Equisatisfiability of the sentences in \ref{enum:EquisatSingleBlock:III} and \ref{enum:EquisatSingleBlock:IV} is obvious---they are even equivalent. It remains to prove equisatisfiability of the sentences in \ref{enum:EquisatSingleBlock:IV} and \ref{enum:EquisatSingleBlock:V}. Let us denote the former by $\varphi_f$ and the latter by $\varphi_R$. 

It is easy to convert any model $\A$ of $\varphi_f$ into a model $\A'$ of $\varphi_R$: simply take the function graphs of the $f_{i}^\A$ as the interpretation of the $R_{i}$ under $\A'$.

Conversely, let $\B'$ be a model of $\varphi_R$. 
We define a structure $\B$ exactly like $\B'$, except for the interpretations of the $f_i$ and $R_i$.
For every $\va \in \fU_\B^{|\vx|}$ we pick one index $\ell \in [m_*]$ such that for every $i = 1, \ldots, |\vy|$ it holds $\<\va, c_{\ell,i}^{\B'}\> \in R_i^{\B'}$, and set $f_{i}^{\B}(\va) := c_{\ell,i}^{\B'}$ (which equals $c_{\ell,i}^{\B}$). Such an index $\ell$ must exists because of $\B' \models \forall \vx. \bigvee^{m_*}_{\ell=1}\bigwedge_{i=1}^{|\vy|} R_i(\vx, c_{\ell,i})$.
To finish the construction of $\B$, we define $R_{i}^\B$ to be the function graph of the newly defined function $f_{i}^\B$ for every index $i$. 
It is easy to check that $\B$ is a model of both $\varphi_R$ and $\varphi_f$.
\end{proof}







\section{Related and Future Work} \label{section:furthrelwork}

\citet{Dreben1979} (page 65) extend the Relational Monadic Fragment to a certain extent in the direction of the BS Fragment 
and call the result \emph{Initially-extended Essentially Monadic Class}. In essence, they allow constants and 
discard the restriction to unary predicate symbols. However, they require that every atom contains at most 
one variable (possibly with multiple occurrences). Consequently, their fragment does not fully include BS.

A broad overview over decidable standard fragments of first-order logic is given in \cite{Borger1997}.
More recent decidability results are often formulated as syntactic restrictions on clause sets~\cite{Weidenbach1999, Fermuller2001, Korovin2013} 
that are incomparable to or subsumed by SF.

There is recent work \cite{Fontaine2009, Wies2009, Charatonik2010, Policriti2012, Bresolin2014} 
which considers the BS(R) Fragment in settings beyond pure first-order logic.
It might be of interest to investigate all those combinations and extensions based on SF instead of the BS(R) Fragment.
Furthermore, SF might also be meaningful to first-order theories over fixed structures such as arithmetic.

Our results on the complexity of satisfiability of SF sentences left some gaps open, which remain to be closed in the future. 
Moreover, the syntactic restrictions on SF sentences may be weakened and still lead to a decidable fragment.







\subsection*{Acknowledgments}
We thank Reinhard P\"oschel for pointing us to Sperner's Theorem, and we thank the anonymous reviewers, who helped in improving the paper by providing valuable hints and suggestions.

Our research was supported by the German Transregional Collaborative Research Center SFB/TR 14 AVACS and the DFG/ANR Project STU 483/2-1 SMArT.

\bibliographystyle{abbrvnat}

\begin{thebibliography}{34}
\providecommand{\natexlab}[1]{#1}
\providecommand{\url}[1]{\texttt{#1}}
\expandafter\ifx\csname urlstyle\endcsname\relax
  \providecommand{\doi}[1]{doi: #1}\else
  \providecommand{\doi}{doi: \begingroup \urlstyle{rm}\Url}\fi

\bibitem[Alagi and Weidenbach(2015)]{Alagi2015}
G.~Alagi and C.~Weidenbach.
\newblock {NRCL} -- {A} {M}odel {B}uilding {A}pproach to the
  {B}ernays--{S}ch{\"{o}}nfinkel {F}ragment.
\newblock In \emph{Frontiers of Combining Systems (FroCoS'15)}, LNCS 9322,
  pages 69--84. Springer, 2015.

\bibitem[Bachmair et~al.(1998)Bachmair, Ganzinger, and Voronkov]{Bachmair1998}
L.~Bachmair, H.~Ganzinger, and A.~Voronkov.
\newblock Elimination of {E}quality via {T}ransformation with {O}rdering
  {C}onstraints.
\newblock In \emph{Automated Deduction (CADE-15)}, LNCS 1421, pages 175--190.
  Springer, 1998.

\bibitem[Baumgartner et~al.(2006)Baumgartner, Fuchs, and
  Tinelli]{ModelEvolutionLemma06}
P.~Baumgartner, A.~Fuchs, and C.~Tinelli.
\newblock Lemma {L}earning in the {M}odel {E}volution {C}alculus.
\newblock In \emph{Logic for Programming, Artificial Intelligence, and
  Reasoning (LPAR-13)}, LNCS 4246, pages 572--586. Springer, 2006.

\bibitem[Baumgartner et~al.(2009)Baumgartner, Fuchs, de~Nivelle, and
  Tinelli]{Baumgartner2009}
P.~Baumgartner, A.~Fuchs, H.~de~Nivelle, and C.~Tinelli.
\newblock Computing finite models by reduction to function-free clause logic.
\newblock \emph{J. Applied Logic}, 7\penalty0 (1):\penalty0 58--74, 2009.

\bibitem[Bernays and Sch{\"{o}}nfinkel(1928)]{Bernays1928}
P.~Bernays and M.~Sch{\"{o}}nfinkel.
\newblock Zum {E}ntscheidungsproblem der mathematischen {L}ogik.
\newblock \emph{Mathematische Annalen}, 99\penalty0 (1):\penalty0 342--372,
  1928.

\bibitem[Bonacina and Plaisted(2014)]{SGGSExposition}
M.~P. Bonacina and D.~A. Plaisted.
\newblock {SGGS} {T}heorem {P}roving: an {E}xposition.
\newblock In \emph{Workshop on Practical Aspects in Automated Reasoning
  (PAAR'14)}, EPiC 31, pages 25--38. EasyChair, 2014.

\bibitem[B{\"{o}}rger et~al.(1997)B{\"{o}}rger, Gr{\"{a}}del, and
  Gurevich]{Borger1997}
E.~B{\"{o}}rger, E.~Gr{\"{a}}del, and Y.~Gurevich.
\newblock \emph{The {C}lassical {D}ecision {P}roblem}.
\newblock Springer, 1997.

\bibitem[Brand(1975)]{Brand1975}
D.~Brand.
\newblock Proving {T}heorems with the {M}odification {M}ethod.
\newblock \emph{{SIAM} J. Comput.}, 4\penalty0 (4):\penalty0 412--430, 1975.

\bibitem[Bresolin et~al.(2014)Bresolin, Della~Monica, Montanari, and
  Sciavicco]{Bresolin2014}
D.~Bresolin, D.~Della~Monica, A.~Montanari, and G.~Sciavicco.
\newblock The light side of interval temporal logic: the
  {B}ernays--{S}ch{\"{o}}nfinkel fragment of {CDT}.
\newblock \emph{Annals of Mathematics and Artificial Intelligence}, 71\penalty0
  (1-3):\penalty0 11--39, 2014.

\bibitem[Charatonik and Witkowski(2010)]{Charatonik2010}
W.~Charatonik and P.~Witkowski.
\newblock On the {C}omplexity of the {B}ernays--{S}ch{\"{o}}nfinkel {C}lass
  with {D}atalog.
\newblock In \emph{Logic for Programming, Artificial Intelligence, and
  Reasoning (LPAR-17)}, LNCS 6397, pages 187--201. Springer, 2010.

\bibitem[Claessen and Lilliestr{\"o}m(2011)]{InfinoxRef11}
K.~Claessen and A.~Lilliestr{\"o}m.
\newblock Automated {I}nference of {F}inite {U}nsatisfiability.
\newblock \emph{J. Autom. Reasoning}, 47\penalty0 (2):\penalty0 111--132, 2011.

\bibitem[Dreben and Goldfarb(1979)]{Dreben1979}
B.~Dreben and W.~D. Goldfarb.
\newblock \emph{The {D}ecision {P}roblem: {S}olvable {C}lasses of
  {Q}uantificational {F}ormulas}.
\newblock Addison-Wesley, 1979.

\bibitem[Ebbinghaus et~al.(1994)Ebbinghaus, Flum, and Thomas]{Ebbinghaus1994}
H.~Ebbinghaus, J.~Flum, and W.~Thomas.
\newblock \emph{Mathematical {L}ogic}.
\newblock Springer, 1994.

\bibitem[Ferm{\"{u}}ller et~al.(2001)Ferm{\"{u}}ller, Leitsch, Hustadt, and
  Tammet]{Fermuller2001}
C.~G. Ferm{\"{u}}ller, A.~Leitsch, U.~Hustadt, and T.~Tammet.
\newblock Resolution {D}ecision {P}rocedures.
\newblock In \emph{Handbook of Automated Reasoning}, pages 1791--1849. Elsevier
  and MIT Press, 2001.

\bibitem[Fontaine(2009)]{Fontaine2009}
P.~Fontaine.
\newblock Combinations of theories for decidable fragments of first-order
  logic.
\newblock In \emph{Frontiers of Combining Systems (FroCoS'09)}, LNCS 5749,
  pages 263--278. Springer, 2009.

\bibitem[Ganzinger and Korovin(2003)]{InstGen03}
H.~Ganzinger and K.~Korovin.
\newblock New {D}irections in {I}nstantiation-{B}ased {T}heorem {P}roving.
\newblock In \emph{Logic in Computer Science (LICS'03)}, pages 55--64. {IEEE},
  2003.

\bibitem[Ge and de~Moura(2009)]{GeMoura09}
Y.~Ge and L.~M. de~Moura.
\newblock Complete {I}nstantiation for {Q}uantified {F}ormulas in
  {S}atisfiabiliby {M}odulo {T}heories.
\newblock In \emph{Computer Aided Verification (CAV'09)}, LNCS 5643, pages
  306--320. Springer, 2009.

\bibitem[Gurevich(1969)]{Gurevich1971}
Y.~Gurevich.
\newblock The {D}ecision {P}roblem for the {L}ogic of {P}redicates and
  {O}perations.
\newblock \emph{Algebra i Logika}, 8:\penalty0 284--308, 1969.

\bibitem[Henkin(1961)]{Henkin1961}
L.~Henkin.
\newblock Some remarks on infinitely long formulas.
\newblock In \emph{Infinistic Methods}, pages 167--183. Pergamon Press, 1961.

\bibitem[Hillenbrand and Weidenbach(2013)]{HillenbrandWeidenbach13}
T.~Hillenbrand and C.~Weidenbach.
\newblock Superposition for bounded domains.
\newblock In \emph{Automated Reasoning and Mathematics -- Essays in Memory of
  William W.\ McCune}, LNCS 7788, pages 68--100. Springer, 2013.

\bibitem[Korovin(2013{\natexlab{a}})]{Korovin13}
K.~Korovin.
\newblock From {R}esolution and {DPLL} to {S}olving {A}rithmetic {C}onstraints.
\newblock In \emph{Frontiers of Combining Systems (FroCoS'13)}, LNCS 8152,
  pages 261--262. Springer, 2013{\natexlab{a}}.

\bibitem[Korovin(2013{\natexlab{b}})]{Korovin2013}
K.~Korovin.
\newblock Non-cyclic {S}orts for {F}irst-{O}rder {S}atisfiability.
\newblock In \emph{Frontiers of Combining Systems (FroCoS'13)}, LNCS 8152,
  pages 214--228. Springer, 2013{\natexlab{b}}.

\bibitem[Letz and Stenz(2002)]{Letz2002}
R.~Letz and G.~Stenz.
\newblock {I}ntegration of {E}quality {R}easoning into the {D}isconnection
  {C}alculus.
\newblock In \emph{Automated Reasoning with Analytic Tableaux and Related
  Methods (TABLEAUX'02)}, LNAI 2381, pages 176--190. Springer, 2002.

\bibitem[Lewis(1980)]{Lewis1980}
H.~R. Lewis.
\newblock Complexity {R}esults for {C}lasses of {Q}uantificational {F}ormulas.
\newblock \emph{J. Comput. Syst. Sci.}, 21\penalty0 (3):\penalty0 317--353,
  1980.

\bibitem[L{\"o}b(1967)]{Lob1967}
M.~L{\"o}b.
\newblock Decidability of the monadic predicate calculus with unary function
  symbols.
\newblock \emph{J. Symb. Logic}, 32:\penalty0 563, 1967.

\bibitem[L{\"o}wenheim(1915)]{Lowenheim1915}
L.~L{\"o}wenheim.
\newblock {\"U}ber {M}{\"o}glichkeiten im {R}elativkalk{\"u}l.
\newblock \emph{Mathematische {A}nnalen}, 76:\penalty0 447--470, 1915.

\bibitem[Nonnengart and Weidenbach(2001)]{NonnengartWeidenbach01handbook}
A.~Nonnengart and C.~Weidenbach.
\newblock Computing {S}mall {C}lause {N}ormal {F}orms.
\newblock In \emph{Handbook of Automated Reasoning}, pages 335--367. Elsevier
  and MIT Press, 2001.

\bibitem[Piskac et~al.(2010)Piskac, de~Moura, and Bj{\o}rner]{Piskac2010}
R.~Piskac, L.~M. de~Moura, and N.~Bj{\o}rner.
\newblock Deciding {E}ffectively {P}ropositional {L}ogic {U}sing {DPLL} and
  {S}ubstitution {S}ets.
\newblock \emph{J. Autom. Reasoning}, 44\penalty0 (4):\penalty0 401--424, 2010.

\bibitem[Policriti and Omodeo(2012)]{Policriti2012}
A.~Policriti and E.~Omodeo.
\newblock The {B}ernays--{S}ch{\"{o}}nfinkel--{R}amsey class for set theory:
  decidability.
\newblock \emph{J. Symb. Logic}, 77:\penalty0 896--918, 2012.

\bibitem[Ramsey(1930)]{Ramsey1930}
F.~P. Ramsey.
\newblock On a {P}roblem of {F}ormal {L}ogic.
\newblock \emph{Proceedings of The London Mathematical Society},
  s2-30:\penalty0 264--286, 1930.

\bibitem[Sperner(1928)]{Sperner1928}
E.~Sperner.
\newblock {E}in {S}atz {\"u}ber {U}ntermengen einer endlichen {M}enge.
\newblock \emph{Mathematische Zeitschrift}, 27\penalty0 (1):\penalty0 544--548,
  1928.

\bibitem[Sturm et~al.(2016)Sturm, Voigt, and Weidenbach]{Voigt2016}
T.~Sturm, M.~Voigt, and C.~Weidenbach.
\newblock {D}eciding {F}irst-{O}rder {S}atisfiability when {U}niversal and
  {E}xistential {V}ariables are {S}eparated.
\newblock In \emph{Logics in Computer Science (LICS'16), to appear}. IEEE/ACM,
  2016.

\bibitem[Weidenbach(1999)]{Weidenbach1999}
C.~Weidenbach.
\newblock {T}owards an {A}utomatic {A}nalysis of {S}ecurity {P}rotocols in
  {F}irst-{O}rder {L}ogic.
\newblock In \emph{Automated Deduction (CADE-16)}, LNCS 1632, pages 314--328.
  Springer, 1999.

\bibitem[Wies et~al.(2009)Wies, Piskac, and Kuncak]{Wies2009}
T.~Wies, R.~Piskac, and V.~Kuncak.
\newblock Combining {T}heories with {S}hared {S}et {O}perations.
\newblock In \emph{Frontiers of Combining Systems (FroCoS'09)}, LNCS 5749,
  pages 366--382. Springer, 2009.

\end{thebibliography}

\end{document}